%% file: rent-sharing-arXiv.tex
\documentclass[11pt]{article}
\usepackage{geometry}
\usepackage{graphicx}
\usepackage{amsmath}
\usepackage{amssymb}
\usepackage[all]{xy}
\usepackage[numbers,square]{natbib}
\newtheorem{theorem}{Theorem}[section]
\newtheorem{corollary}[theorem]{Corollary}
\newtheorem{lemma}[theorem]{Lemma}

\newtheorem{definition}[theorem]{Definition}
\newtheorem{remark}[theorem]{Remark}

\newtheorem{proposition}[theorem]{Proposition}

\def\squarebox#1{\hbox to #1{\hfill\vbox to #1{\vfill}}}
\newcommand{\qed}{\hspace*{\fill}\vbox{\hrule\hbox{\vrule\squarebox{.667em}\vrule}\hrule}\smallskip}
\newenvironment{proof}{\noindent{\bf Proof:~~}}{\(\qed\)}

\newcommand{\longver}[1]{}

\usepackage{amsmath}
\usepackage{amssymb}
\usepackage{array}

\newcommand{\ignore}[1]{}
\newcommand{\shortversion}[1]{}

\newcommand{\rentalprob}{{Shared--Rental problem}}
\newcommand{\lexmaxws}{{\emph{lex-max-WS}}}

\newcommand{\lexmax}{{\lexmaxws}}

\begin{document}
% Title portion. Note the short title for running heads
%\title{A New Approach to Fair Distribution of Welfare}
%\author{Submission 125}

% note that the abstract must come before \maketitle

\title{A New Approach to Fair Distribution of Welfare}
%\titlerunning{Bottleneck Models}

\author{
	Moshe Babaioff
	\thanks{
		Microsoft Research, moshe@microsoft.com.
	}
	\and{
		Uriel Feige
		\thanks{
			Weizmann Institute, Rehovot, Israel, Uriel.Feige@weizmann.ac.il. Supported in part by the Israel Science Foundation (grant No. 1388/16).
			Part of this work was done at Microsoft Research, Herzeliya.
	}}
}
\maketitle
\begin{abstract}
	We consider transferable-utility profit-sharing games that arise from settings in which agents
	need to jointly choose one of several alternatives, and may use transfers to redistribute the welfare
	generated by the chosen alternative. One such setting is the \rentalprob, in which students jointly
	rent an apartment and need to decide which bedroom to allocate to each student, depending on
	the student's preferences. Many solution concepts have been proposed for such settings, ranging from mechanisms without transfers, such as Random Priority and the Eating mechanism, to mechanisms with transfers, such as envy free solutions, the Shapley value, and the Kalai-Smorodinsky bargaining solution. We seek a solution concept that satisfies three natural properties, concerning efficiency, fairness and decomposition. We observe that every solution concept known (to us) fails to satisfy at least one of the three properties. We present a new solution concept, designed so as to satisfy the three properties. A certain submodularity condition  (which holds in interesting special cases such as the Shared-Rental setting) implies both existence and uniqueness of our solution concept.   
	%is based on the Bondareva-Shapley theorem, whereas uniqueness follows from a theorem of Dutta and Ray about Lorenz domination.
\end{abstract}

% Renew this after \maketitle if the default list of authors is too long for headers
%\renewcommand{\shortauthors}{W.\ Vickrey et.\ al.}

\input{content}

\subsection*{Acknowledgments}

We thank Herve Moulin and Eyal Winter for helpful discussions.

% Bibliography
%%%\bibliographystyle{ACM-Reference-Format}
%\bibliographystyle{abbrvnat}

\bibliography{bib}

\appendix

\input{the-appendix}

%\input{old-sections}

%\begin{acks}
%This is where we put acknowledgments.
%\end{acks}

%\input{solutionConcepts}

%\input{app-algorithms}
\end{document}

%% file: content.tex
\section{Introduction}
\subsection{Background}
\label{sec:background}
% \mbedit{this is an edit}\mbcomment{this is a comment}
We introduce a new solution concept for situations in which agents with cardinal preferences need to jointly choose one alternative from a set of alternatives, possibly compensating each other using transfers.
This is a well studied setting in cooperative game theory, and we follow a normative approach that specifies properties that we wish our solution concept to have, and then design a solution concept that meets these specifications. To motivate our new solution concept and contrast it with well established previous solution concepts, we start with an example.

Suppose that three students jointly rent a three bedroom apartment for a total rent of $r$ units of money. The students need to do two things. One is to jointly pay the rent, and the other is to solve the allocation problem, namely, decide which student gets which room, possibly compensating each other with money.
We assume that the students are equals, in the sense that each student bears equal responsibility in paying the rent, and equal eligibility in receiving a room. Being equals, each student first pays $r/3$ towards the rent. It remains to solve the allocation problem, where this solution may possibly involve transfer of money among the students.
%Namely, the students are jointly committed to pay rent regardless of whether they reach agreement on the allocation (or even regardless of whether they move into the apartment at all), and hence they each pay $r/3$ (as they have equal responsibility) towards the rent, prior to addressing the allocation problem.

\begin{remark}
In cases in which no student receives a transfer larger than $r/3$, transfers may be implemented indirectly by having students pay unequal parts of the rent.
However, in this paper we do not constrain transfers to be smaller than $r/3$, and the question of whether transfers are implemented as direct transfers among the students or as modification to rent payments is not a concern of the current paper.
%OLD: In cases in which no student receives a transfer larger than $r/3$, transfers may be implemented indirectly by having students pay unequal parts of the rent. The question of whether transfers are implemented as direct transfers among the students or as modification to rent payments is not a concern of the current paper. We do not constrain transfers to be smaller than $r/3$, and we may invoke transfers even if $r=0$.
\end{remark}

A common approach for allocating rooms (and other goods) is using the {\em Random Priority} mechanism (a.k.a. {\em random serial dictatorship}), that we abbreviate as RP. A total order among the students is chosen uniformly at random, and each student in her turn chooses a room among those that are still available. RP has obvious advantages, being easy to implement in practice, agents (students in our case) have dominant strategies (given an agent's turn to choose, she should simply choose the available alternative that she most prefers), and being perceived as ``fair'' (all agents are treated equally from the mechanism's point of view). A significant drawback of RP is that it does not maximize welfare -- the resulting allocation may produce less welfare (sum of utilities) than alternative allocations. Hence some economic efficiency is lost.

Let us consider a concrete example.
Suppose that the students can express their valuation for rooms in units of money,
and that they are risk neutral (they wish to maximize the expected received value). Suppose further that for some small $0 < \delta < \frac{1}{4}$, the value that each student derives by being given each of the rooms is as in the following table:

\renewcommand{\arraystretch}{1.2}

\begin{center}
\begin{tabular}{|l||c|c|c|}
  \hline
  % after \\: \hline or \cline{col1-col2} \cline{col3-col4} ...
  Example 1 & {\bf Room 1} & {\bf Room 2}  & {\bf Room 3} \\ \hline \hline
  {\bf Student 1} & $1 - \delta$ & $\delta$ & 0 \\ \hline
  {\bf Student 2} & $1 - 2\delta$ & $2\delta$ & 0 \\ \hline
  {\bf Student 3} & 0 & $\frac{1}{2} - \delta$ & $\frac{1}{2} + \delta$ \\
  \hline
\end{tabular}
\end{center}

\medskip

The maximum welfare allocation assigns room~$i$ to student~$i$ for every $i$, giving welfare of $\frac{3}{2} + 2\delta$. However, RP will result with probability half with an assignment in which student~2 gets room~1, giving welfare $\frac{3}{2}$,
and hence the expected welfare of RP is $\delta$ lower than optimal.

The RP mechanism does not involve transfer of money among agents. In our language, we refer to it as an ANT, which is an abbreviation for Allocation mechanism with No Transfers. To overcome its weaknesses (shared by other ANTs as well), one often considers allocation mechanisms with transfers (abbreviated as AWT -- Allocations With Transfers). AWTs allow for the following paradigm: first choose a maximum welfare allocation (thus creating
the largest pie to divide: the maximum possible
welfare to distribute among the agents), and then employ monetary transfers among the agents so as to distribute the high welfare to all agents, so as to satisfy some fairness criteria. In the example above, this would mean assigning room~$i$ to student~$i$ for every $i$, and then figuring out what the transfers should be so that the combination of allocation with transfers would be ``fair''.

To reason about transfers, we make the assumption that students have quasi-linear utilities: the utility of a student is simply the sum of her value for the room that she receives plus the transfer that she receives (the transfer may be negative if the student gives money rather than receives money).
Moreover, we assume that the mechanism that computes the allocation and the transfers has access to the true valuations of the students.
%Moreover, we require students to disclose their valuation functions to the mechanism that computes the allocation and the transfers. We further assume that they do so truthfully. 
(This full information assumption is standard in cooperative game theory, and there are impossibility results showing that it cannot be avoided in our setting. See more details in Appendix~\ref{sec:discuss}.)
Within such a setting, there is a well studied class of AWTs that is referred to as {\em envy free} solutions~\cite{Foley67,Su99,GMPZ}). The basic principle is that one associates a transfer with each room (where the sum of transfers equals~0 -- this is a {\em budget balance} condition) such that given the transfers, each student (weakly) prefers a different room. Then each student gets the room and associated transfers that she prefers, and no one prefers to switch with another agent.
In the example above we can associate the following transfers with the rooms:

\renewcommand{\arraystretch}{1.2}

\begin{center}
\begin{tabular}{|c|c|c|}
  \hline
  % after \\: \hline or \cline{col1-col2} \cline{col3-col4} ...
    {\bf Room 1} & {\bf Room 2}  & {\bf Room 3} \\ \hline
  $-\frac{2}{3} + 2\delta$ & $\frac{1}{3} - \delta$ & $\frac{1}{3} - \delta$ \\
  \hline
\end{tabular}
\end{center}

\medskip

These transfers are indeed budget balanced and envy free, that is,
each student~$i$ prefers her assigned room~$i$ (along with the associated payment) over any other room,
leading to an allocation that maximizes welfare and in which supposedly every student is happy (as she got her most preferred one out of the three available options).

Let us consider a natural question. Suppose that the students initially intend to use the RP mechanism.  Will the students be better off by using the
envy free mechanism (that we abbreviate EF) instead of using RP? In some respects, the answer is no: RP is simpler to implement than EF, as it does not require students to disclose their valuation functions and to implement transfers. In other respects the answer is yes: EF generates higher welfare. But let us consider this last aspect more carefully. The social justification to maximize welfare is (in our opinion) the belief that the extra welfare will eventually get distributed to all members of the society that contributed to the increase in welfare. Is it the case that the increase in welfare (generated by moving from RP to EF) is distributed over the three students in a reasonable way? The answer is negative in our opinion.

\begin{itemize}

\item In RP, the sum of expected values derived by students~1 and~2 is~1. In EF, the sum of values increases to $1 + \delta$, but the sum of what they lose due to transfers is $\frac{1}{3} - \delta$. If $\delta < \frac{1}{6}$, each of the two students gets higher expected utility from RP than from EF. It is not true that the increase in welfare is distributed over all students in a way that \emph{every} student (at least weakly) benefits.

\item Student~3 contributes nothing to the increase in welfare when changing from RP to EF (in both cases her allocation is exactly the same -- room~3). Nevertheless, under EF, student~3 not only gets her most preferred room, but also gets paid. Moreover, this payment is even larger than the total increase in welfare that EF offers compared to (the expected welfare of) RP.

\end{itemize}

Another aspect that we find troublesome with the EF solution is the following.

\begin{itemize}

\item In every Pareto efficient allocation, student~3 gets room~3 and the only question is which of the rooms~1 and~2 is allocated to which of the students~1 and~2. Hence the instance naturally decomposes into two subinstances, $I_3$ involving room~3 and student~3, and $I_{1,2}$ involving the other two students and two rooms. If one does this decomposition and then employs an EF mechanism on each component separately, student~3 does not receive any payments from the other students, and hence the resulting payments are completely different from those without the decomposition. Likewise, suppose that we had started with two separate instances, $I_{1,2}$ and $I_3$ as above, where every student prefers the rooms in her own instance over those in the other instance (it may even be that each instance concerns a different apartment). If we use EF mechanisms, then combining the two instances into one results in different payments compared to solving each of the instances separately.  This sensitivity of the payments in EF mechanisms to composition and decomposition of instances (importantly, we are considering here cases in which composition and decomposition have no effect on the allocation itself) may lead to disagreements among the agents regarding what constitutes a single instance.   
%\mbcomment{Are you suggesting removing the rest?} Likewise, had we started with an instance $I_{1,2}$ involving two rooms in South-Africa and $I_3$ involving one room in Venezuela (and no student wishes to change country), then by thinking of them as one instance the EV mechanism would have us transfer money from South Africa to Venezuela (envy has no borders).

\end{itemize}

An allocation instance may have several different envy free solutions, but the above shortcomings are shared by \emph{all} envy free solutions in the above example, provided that $\delta$ is sufficiently small\footnote{If $\delta$ is sufficiently small, then in {\em every} envy free solution the transfer associated with room~3 is positive: Let $p_i$ denote the transfer associated with room $i$. For student~3 not to envy student~2 we must have $p_2 \le p_3 + 2\delta$. For student~2 not to envy student~1 we must have $p_1 \le p_2 - 1 + 4\delta \le p_3 - 1 + 6\delta$.  Together with the budget balance requirement we have that $p_3 = -p_1 -p_2 \ge -p_3 + 1 - 6\delta -p_3 - 2\delta$, implying that $p_3 \ge \frac{1 - 8\delta}{3}>0$, where the last inequality holds  when $\delta$ is sufficiently small.}.

Summarizing, ANT mechanisms such as RP need not maximize welfare. AWT mechanisms can address this weakness. A common AWT approach, that of envy free (EF) mechanisms has elegant conceptual properties when considered in isolation. However, when comparing its outcomes to those of RP, we identified several troubling aspects with its transfers. These include the fact that despite increase in welfare (compared to RP), some individual agents suffer loss in (expected) utility, the fact that agents who contribute nothing to the increase in welfare might receive payments (even beyond the total increase in welfare), and the fact that natural composition and decomposition properties are not respected by EF mechanisms.

Many other AWT approaches that have been proposed in the literature can be applied in the room allocation setting. They include (among others) the {\em Shapley value}, the {\em Nucleolus}, the {\em Nash bargaining} solution and the {\em Kalai-Smorodinsky (KS) bargaining} solution. Every one of them suffers from at least one of the troubling aspects listed above, see sections~\ref{sec:shapley}, \ref{sec:compare-rental} and Appendix  \ref{appendix:bargaining} for more details. In fact, the same holds for every AWT approach that we could find in the literature. Hence despite the many solution concepts that already exist, we find it appropriate to introduce a new AWT mechanism that does not suffer from any of the troubling aspects listed above.

% \section{Overview}

\subsection{The model}

We consider {\em transferable-utility profit-sharing games}, a setting that has been studied in previous work (e.g., by Moulin~\cite{Moulin85}). The room allocation problem of the previous section is a special case of this more general setting.

There is a set $\cal{N}$ of $n$ {\em agents} (also referred to as {\em players}) and a set $\cal{A}$ of {\em alternatives}. Every agent $i \in \cal{N}$ has a valuation function $v_i : {\cal{A}} \rightarrow \mathbb{R}$. All valuation functions are expressed in the same units (of money). We let $v = (v_1, \ldots, v_n)$ denote the tuple of all the valuation functions.
%Given $v$, it might be that some alternatives are considered not to be feasible. The set of alternatives feasible under $v$ is denoted by ${\cal{A}}$, with ${\cal{A}} \subseteq{\cal{A}}$.
An NT ({\em no transfers}) {\em social choice function} $f$ receives as input the pair $({\cal{A}}, v)$ that includes the set of alternatives and the valuation functions, and outputs one of the alternatives from ${\cal{A}}$.  A randomized NT social choice function may use randomization when choosing its output. Consequently, its output is a probability distribution over alternatives.

Given the tuple $v$ of valuation functions, a set $S \subseteq{\cal{N}}$ of agents and an alternative $A \in {\cal{A}}$, the {\em welfare} $w_{S,v}(A)$ that alternative $A$ offers to $S$ is defined as $w_{S,v}(A) = \sum_{i \in S} v_i(A)$. An NT social choice function $f$ {\em maximizes welfare} (with respect to $\cal{N}$) if the alternative $A^* \in {\cal{A}}$ that $f$ selects satisfies $w_{{\cal{N}},v}(A^*) \ge w_{{\cal{N}},v}(A)$ for all $A \in {\cal{A}}$.

We allow transfer of money among agents. Such transfers are represented as a vector $p = (p_1, \ldots, p_n)$, where $p_i$ is the payment to agent $i$, measured in units of money. We refer to the case of $p_i > 0$ as an {\em in-payment} (the amount of money of agent $i$ increases), and to the case of $p_i < 0$ as an {\em out-payment} (the amount of money of agent $i$ decreases). A transfer vector $p$ is {\em budget balanced} if $\sum_{i=1}^n p_i = 0$. A {\em transfer function} $g$ receives as input the triple $({\cal{A}}, v, A^*)$ that includes the set of alternatives, the valuation functions, and an alternative chosen by an NT social choice function, and outputs a budget balanced transfer vector.

We assume that the utility functions of the agents are {\em quasi-linear}.
Namely, for agent $i \in {\cal{N}}$ with valuation function $v_i$, her utility $u_i$ from the pair of alternative $A$ and transfer vector $p$ is $u_i(A,p) = v_i(A) + p_i$. We further assume a setting of ``full information upon request": the social planner may request information about valuation functions of agents (this information might be limited to the ordinal preferences of an agent over a set of alternatives, or might be as general as the full valuation function of an agent), and the agents reply truthfully to such requests. (These assumptions will be discussed in Appendix~\ref{sec:discuss}.)

%Mathematically, a setting of full information upon request is equivalent to a full information setting in which the vector $v$ of valuation functions is known upfront. However, conceptually there is a difference between the two settings. The ``upon request" setting is meant to remind the reader that there is a cost incurred by the agents in providing information about their valuation functions to the mechanism. We only note that such a cost exists (e.g., due to computational overhead for the agents, communication overhead, compromise of privacy), and that presumably it increases as the amount of information requested increases. We do not attempt to quantify this cost, and it does not enter the quantitative aspects of our mechanisms. Rather, the existence of such a cost serves as a guiding principle for designing mechanisms: do not request detailed information from an agent unless the information is used for the agent's benefit.

Let us illustrate how the above model captures the example presented in Section~\ref{sec:background}, of three students renting a three bedroom apartment. $\cal{N}$ corresponds to the set of three students, and $\cal{A}$ corresponds to the set of six possible permutations over rooms, matching one room to one student. The valuation functions $v_i$ are as in the example.
%Given the tuple $v$ of valuation functions, the set of feasible alternatives ${\cal{A}}$ might be smaller than the set of all alternatives. For example, it makes sense to discard those alternatives that are not Pareto optimal (e.g., those alternatives in which student~3 gets room~1).
An example of a randomized NT social choice function is the output of the {\em Random Priority} (RP) mechanism: once $v$ is given (in fact, knowledge of ordinal preferences suffices here) RP induces a well defined probability distribution over alternatives. The envy-free allocation and transfers provided in the example are a solution (implicitly) involving an NT social choice function $f$ and a transfer function $g$.

\subsection{Our contribution}

For the setting described above, we wish to design a solution concept that has two components: an NT social choice function, and an associated transfer function. We have three goals. One is {\em economic efficiency}.  This goal is easily attainable in our full information framework -- we simply select a welfare maximizing alternative, which we denote by $A^*$. (If there are several welfare maximizing alternatives, $A^*$ denotes one of them, selected arbitrarily.)  Another goal is to achieve {\em fairness}, in the sense that the welfare will be shared ``fairly'' among all agents. Achieving this goal is made possible by the use of transfers.  Those agents for which alternative $A^*$ is undesirable can be compensated by in-payments, and the budget balance requirement can be met by extracting an equal amount of out-payments from those agents who do desire alternative $A^*$. The assumption that agents have quasilinear utility functions simplifies the accounting of the extent to which utility derived from payments can replace utility derived from the selected alternative. The third goal is that of {\em decomposability}, which basically means that if a large game involving multiple agents can be naturally decomposed into many smaller games over disjoint sets of agents, then the solution of the large game should also decompose into solutions of the smaller games. Equivalently, one should be able to solve each smaller game separately, and obtain a solution to the large game as the concatenation of the solutions to the smaller games.

Our contributions in this work are in setting the above three goals, proposing definitions for the fairness properties and decomposition properties that they refer to, proposing a solution concept that attains the above three goals, and providing sufficient conditions for its existence and uniqueness. Here is an informal statement of our main result when specialized to the \rentalprob.

{\bf Theorem (informal)}. {\em The \lexmaxws\ solution (introduced in our work) for the \rentalprob\  maximizes welfare and satifies the fairness and the decomposition properties alluded to above (and formally defined later in this paper). Moreover, in a well defined sense, it is the unique solution that satisfies these properties.}

Here are more details regarding our contributions:

\begin{enumerate}

\item We propose a new notion of fair solutions, the {\em welfare-sharing core} (abbreviated WS-core). See Definition~\ref{def:WS}. It combines three principles that are briefly sketched below.

    \begin{enumerate}

    \item One principle is  {\em domination} with respect to the utility agents can
receive from a {\em disagreement point}, or {\em reference point}.
%The domination aspect by itself is not new (e.g., it explicitly appears in~\cite{Moulin85}).
This is mathematically similar to the familiar concept of {\em individual rationality} (IR), though conceptually there is a distinction between these two notions.
See Section~\ref{sec:domination} for more details.

    \item Another principle is that fairness entails not only lower bounds on the utilities that agents derive from the solution, but also natural upper bounds. We introduce a set-function $W_{max}$, where for a set $S$ of agents, $W_{max}(S)$ is the welfare that $S$ could derive from the alternative that is best for $S$. The same notion appears  in~\cite{Moulin92}, where is is referred to as {\em stand alone} utility. We require that the utility that a solution (with transfers) offers to a set $S$ of agents does not exceed $W_{max}(S)$.  This leads to the notion that we (and~\cite{Moulin92}) refer to as the {\em anticore}.  See Section~\ref{sec:uppercore} for more details.

\item Another principle is that of {\em decomposability} , as discussed above (see Section~\ref{sec:decompose} for more details).
%Basically it says that if a large game involving multiple agents can be naturally decomposed into many smaller games over disjoint sets of agents, then the solution of the large game should also decompose into solutions of the smaller games. Equivalently, one should be able to solve each smaller game separately, and obtain a solution to the large game as the concatenation of the solutions to the smaller games.
A key property of the anticore is that it decomposes: the anticore of a decomposable game is the concatenation of the anticores of each of the component games.

\end{enumerate}

\item We show that in our setting, if $W_{max}$ is submodular, then the WS-core is non-empty. See Theorem~\ref{thm:BS}.

\item We propose  to use egalitarian considerations (specifically, the  lexicographically-maximal welfare-sharing rule, denoted \lexmaxws) for selecting a single solution from the WS-core, see Section \ref{sec:select}.
When $W_{max}$ is submodular,
we show (see Theorem~\ref{thm:DR89}, {which relates to a previous result of Dutta and Ray~\cite{DR89}}) that different egalitarian considerations (e.g., also the min-square rule, {defined in Section \ref{sec:select}}) all lead to the same unique solution.

\item When $W_{max}$ is submodular, we show that computing the \lexmaxws\ solution %is NP-hard in general, but
can be done in polynomial time.  See Appendix~\ref{sec:algorithms}.
Moreover, it is a continuous function (with a small Lipschitz constant) of the valuation functions at points where the disagreement utility is a continuous function of the valuations. %when $W_{max}$ is submodular.
See Appendix~\ref{sec:continuity}.
The \lexmaxws\ solution may not be continuous at points in which the disagreement utility is not continuous.

\item We explain the similarities and differences between our new solution concept and several related notions.
These include coalitional games and imputations (Section \ref{sec:discussion}); {\em cost-sharing} games  (Section \ref{sec:discussion}); notions related to our notion of decomposability, such as {\em Separability} (Section~\ref{app:separable}) and {\em consistency for reduced games} (implicitly addressed in Section~\ref{app:nucleolus});
% \item Solutions that are {\em reasonable from above} (REAB). See Section \ref{sec:discussion}.
previous notions referred to as the {\em anticore} (Section~\ref{sec:uppercore});
{\em egalitarian} solution concepts and {\em Lorenz ordering} (Section \ref{sec:select});
the {\em Shapley value} (Section \ref{sec:shapley}); the {\em Nucleolus} (Appendix \ref{appendix:bargaining}); {\em envy free} solutions (Section \ref{sec:Envy-free solutions}); {\em Nash bargaining} and {\em Kalai-Smorodinsky (KS) bargaining} (Appendix \ref{appendix:bargaining}); {\em population monotonicity} and {\em resource monotonicity} (Appendix~\ref{sec:monotonicity}).

\item We show that for the \rentalprob\ $W_{max}$ is submodular, and hence the \lexmaxws\ solution
enjoys those properties shown above to be implied by submodularity. In addition, the \lexmaxws\ solution dominates \emph{Random Priority} (by definition), and moreover, when instances are ``decomposable" it satisfies a strong notion of decomposability. See Section \ref{sec:rental}.

%\item When presenting a new solution concept, it is important to note not only the desirable properties that it has, but also those that it lacks. We provide examples in which {\em population monotonicity} and {\em resource monotonicity} do not hold for the \lexmaxws\ solution for the \rentalprob, when the disagreement mechanism is  \emph{Random Priority}. See Appendix~\ref{sec:monotonicity}.

\end{enumerate}

\section{The Welfare-Sharing Core}\label{sec:core}

Our starting point is the (not necessarily new) premise that statements such as ``this solution is fair'' have no rigorous meaning on their own. Rather, the fairness of a solution needs to be judged in relation to a {\em reference context}. In our definition of fairness, the reference context will be the set ${\cal{A}}$ of alternatives together with a probability distribution $\pi$ over ${\cal{A}}$ (which we will refer to as a {\em reference point}, or {\em disagreement point}). We now present the definition of the WS-core, and then follow it up with a discussion and comparison with related work.

A {\em solution} $(A^*,p)$ is composed of a welfare maximizing alternative $A^*$ and a budget balanced transfer vector $p = (p_1, \ldots, p_n)$. The {\em utility} that agent $i$ derives from solution $(A^*,p)$ is $u_i(A^*,p) = v_i(A^*) + p_i$. In our context, two solutions $(A^*,p)$ and $(A'^*,p')$ are {\em equivalent} if  $u_i(A^*,p) = u_i(A'^*,p')$ for every agent $i$. Consequently, we sometimes refer to the utility vector $\left(u_1(A^*,p), \ldots, u_n(A^*,p)\right)$ as the solution.

A solution will need to satisfy certain constraints, where these constraints are expressed as a function of the utilities that agents derive from the solution. We shall use $w_{S,v}(A) = \sum_{i\in S} v_i(A)$ to denote the welfare derived by a set $S$ of agents from an alternative $A$, and $u_S(A^*,p) = \sum_{i\in S} u_i(A^*,p)$ to denote the utility derived by $S$ from solution $(A^*,p)$.
%and $W_{min}(S) = \min_{A \in {\cal{A}}}[\sum_{i\in S} v_i(A)]$ indicating the minimum welfare achievable by $S$.

We associate two classes of constraints with solutions $(A^*,p)$:

\begin{enumerate}

\item {\bf Domination:} We assume that a probability distribution $\pi$ over ${\cal{A}}$ is given, where $\pi(A)$ denotes the probability associated with alternative $A$. This distribution represents the alternative that would be chosen in the absence of agreement to use a mechanism with transfers. As such, the distribution $\pi$ may depend on the valuations $v$, and we shall sometimes use the notation $\pi_v$ to make this explicit. The value that agent $i$ derives from $\pi_v$ is $\sum_{A \in {\cal{A}}} \pi_v(A)v_i(A)$, and we refer to it as the agent's disagreement utility. The domination constraints require that $u_i(A^*,p) \ge \sum_{A \in {\cal{A}}} \pi_v(A)v_i(A)$ holds for every agent $i$.

\item {\bf The anticore:}  We introduce a welfare function over sets of agents, which we denote by $W_{max}$.  For every $S \subseteq{\cal{N}}$ let $W_{max}(S) = \max_{A \in {\cal{A}}}[\sum_{i\in S} v_i(A)]$ indicate the maximum welfare achievable by $S$. The anticore constraints require that $u_S(A^*,p) \le W_{max}(S)$ for every set $S \subseteq{\cal{N}}$.

\end{enumerate}

\begin{definition}[WS-core]
\label{def:WS}
Suppose one is given a tuple $v$ of valuation functions, a set ${\cal{A}}$ of alternatives, and a probability distribution $\pi_v$ over ${\cal{A}}$. A solution $(A^*,p)$ (composed of an alternative $A^* \in {\cal{A}}$ that maximizes welfare and a budget balanced vector $p$ of transfers) is said to belong to the {\em welfare-sharing core} (WS-core) if the solution $(A^*,p)$ satisfies the above two sets of constraints (domination and anticore) with respect to the given $v$ and $\pi_v$.
\end{definition}

%As we shall see, the lower core constraints are redundant, as they are implied by the anticore constraints, as well as by the domination constraints. Hence the lower core constraints play no mathematical role, but they do help clarify our solution concept and its relation to other solution concepts.

There are cases in which the WS-core is empty. Here is one such example. Suppose that there are three agents and two alternatives. $A_1$ is the disagreement alternative and all agents value it as~0, whereas $A_2$ is the alternative that maximizes welfare, agent~1 values it as $-1$, whereas each of the other two agents values it as~$1$. Hence there is welfare of $-1+1+1 = 1$ to share among the three agents, and each agent needs to receive utility at least~0 (his disagreement utility). The function $W_{max}$ has value~0 both for the set $\{1,2\}$ and for the set $\{1,3\}$, and hence there is no way of sharing the welfare without violating at least one of the anticore constraints.

Despite the above, in important special cases, the WS-core is nonempty. We first recall some standard terminology. A set function $f$ is {\em monotone} if $f(S) \ge f(T)$ for all $T \subset S$. A set function $f$ is {\em submodular} if for every two sets $S$ and $T$ it holds that $f(S) + f(T) \ge f(S \cap T) + f(S \cup T)$. Equivalently, $f$ is submodular if it has the decreasing marginal returns property:  for every item $i$ and two sets $S \subset T$ it holds that $f(S\cup\{i\}) - f(S)  \ge f(T \cup \{i\}) - f(T)$. A submodular function need not be monotone.

Our main existence result is the following:

\begin{theorem}
\label{thm:BS}
Given a tuple $v$ of valuation functions, a set ${\cal{A}}$ of alternatives, and a probability distribution $\pi$ over ${\cal{A}}$, either one of the following conditions suffices in order for the WS-core to be nonempty.
\begin{enumerate}
\item $W_{max}$ is submodular (though not necessarily monotone).
\item $W_{max} - W_{\pi}$ is monotone (though not necessarily submodular), where $W_{\pi}(S) = \sum_{i\in S}\sum_{A \in {\cal{A}}} \pi_v(A)v_i(A)$ is the expected value derived by set $S$ from the disagreement distribution $\pi$. Note: if the disagreement utilities are~0, then a sufficient condition (though not necessary) for $W_{max} - W_{\pi}$ to be monotone is that the valuation functions are nonnegative.
\end{enumerate}
\end{theorem}

The proof of Theorem~\ref{thm:BS} appears in Appendix~\ref{sec:main-thm-proof}. In is based on the following approach.
Similar to proofs of the well known Bondareva-Shapley theorem~\cite{Bondareva63,Shapley67}, non-emptiness of the WS-core can be cast as a feasibility question for a certain linear program, which then translates to showing that the dual of the linear program is bounded. Each of the submodularity and monotonicity conditions listed above is shown to imply that the dual is bounded, thus proving the theorem.

In Section~\ref{sec:discussion} we provide more details on the domination and anticore constraints. We now formally define the decomposability property that plays an important part in our work.

\section{Decomposability}
\label{sec:decompose}

A central aspect in applying game theory, social choice and mechanism design in practice is that of decomposing large games into smaller games and reasoning about each small game separately. One may view all humanity (and other strategic living creatures) as participating in one huge game in which individuals pursue their own goals and have multiple interactions with other individuals. This game is too heterogeneous and complicated to reason about as a whole. Moreover, the actions of some individuals have very low influence (if at all) on some other individuals, to the extent that they can be ignored. Thus, to be able to reason about interactions between individuals, it is reasonable to decompose this huge game into smaller games, involving smaller numbers of individuals, and having a more homogeneous character. For example, a smaller game might be a particular auction, a particular room allocation problem, or elections for a particular position. In such a smaller game we specify who the players are, what actions are available to them, what the possible outcomes are, and assume that the value derived by the players from the game depends only on the outcome of that game. The decomposition of the huge ``game of humanity" into smaller games is a modeling decision that captures reality only in some approximate sense (the small games are not really isolated from each other, there might be players affecting or affected by the game that we are not aware of, etc.), but seems to be an unavoidable modeling decision in areas such as social choice and mechanism design.

Given the ubiquity of game decompositions, we think it is important that mechanisms (for profit-sharing games, in our context) will remain consistent throughout decompositions. Ideally, we would like the solution to problems that have a natural partition to subproblems, to be the same whether or not we consider the problem as a whole and find a solution, or consider each subproblem separately and find a solution to each.

Motivated by the above view,  in this section we introduce formal definitions for the notion of an instance being decomposable, and for two notions of decomposability for mechanisms: weak and strong.
% the definitions that are presented in this section.

%that mechanisms (for profit sharing games, in our context) will remain consistent throughout decompositions, in a sense that we define next.

%We call this notion \emph{strong decomposability}. We also define a somewhat weaker notion that we call \emph{weak decomposability} that  only requires that no transfers are made between "set of agents involved in separate games", yet does not require the solutions to be independent of the other games.

Let $\cal{A}$ be a set of alternatives, $\cal{N}$ be a set of agents, and let $v = (v_1, \ldots, v_n)$ be a tuple specifying the valuation functions of the agents. We say that alternative $A \in {\cal{A}}$ is {\em Pareto optimal} with respect to a set $S \subset {\cal{N}}$ of agents if for every alternative $B \in {\cal{A}}$, either there is some agent $i \in S$ such that $v_i(A) > v_i(B)$, or for all agents $i \in S$ it holds that $v_i(A) = v_i(B)$.

\begin{definition}[independent component, decomposable instance]
\label{def:decompose}
A set of players $S \subset {\cal{N}}$ is referred to as an {\em independent component} (or just {\em component}, for brevity) if for every alternative $A \in {\cal{A}}$ that is Pareto optimal with respect to $S$ (given $v$) and for every alternative $B \in {\cal{A}}$ that is Pareto optimal with respect to ${\bar{S}} = {\cal{N}} \setminus S$, there is an alternative $C \in {\cal{A}}$ (possibly $C=A$ or $C=B$) such that for every agent $i \in S$ it holds that $v_i(C) = v_i(A)$, and for every agent $j \in {\cal{N}} \setminus S$ it holds that $v_j(C) = v_j(B)$. We say that an instance is {\em decomposable} if it has a component that is nontrivial (the component is neither empty, nor the whole instance).
\end{definition}

It is implicit in the above definition that if a decomposable instance has more than one Pareto optimal alternative, then there are agents that are indifferent among some choices of alternatives.

Observe that if $S \subset {\cal{N}}$ is a component then so is ${\cal{N}} \setminus S$. Definition~\ref{def:decompose} implies that if each of the two components $S$ and ${\cal{N}} \setminus S$ selects a most preferred alternative on its own (such an alternative will be Pareto optimal with respect to the component), then there will be no conflicts between the two choices -- we will be able to select a single alternative that is just as good, from the point of view of every player in every component.

As an example to the decomposition concept introduced above, consider the \rentalprob\  example from Section~\ref{sec:background}, with valuation functions as in the table titled Example~1, and with $\delta < \frac{1}{4}$. In that example, there are two components, one containing  Students~1 and~2, and the other containing Student~3. Every alternative $A$ that is Pareto optimal for the first component assigns the first two rooms to the first two students,
%(in one of the two possible orders if $\delta < \frac{1}{4}$, and in reverse order of $\frac{1}{4} < \delta < \frac{1}{2}$),
and every alternative $B$ that is Pareto optimal for the second component assigns the third room to the third student. The two alternatives $A$ and $B$ can be replaced by one alternative $C$ (in fact, in this simple example it will hold that $C = A$ as there is only one room in the second component), and every agent values $C$ as being equally good as the alternative chosen by his own component.

% Long version (This was moved to appendix \ref{app:decomp})
\longver{
The above discussion is complemented by the following proposition.

\begin{proposition}
\label{lem:welfarecomponent}
Any alternative that maximizes welfare also maximizes welfare for each component separately.
\end{proposition}

\begin{proof}
Let $A$ be an alternative that maximizes welfare, namely, for which $\sum_{i\in {\cal{N}}} v_i(A)$ is largest possible. For a component $S$, let $B$ be an alternative that maximizes the welfare of $S$, namely, for which $\sum_{i\in S} v_i(B)$ is largest possible. We need to show that $\sum_{i\in S} v_i(A) = \sum_{i\in S} v_i(B)$.

Suppose for the sake of contradiction that $\sum_{i\in S} v_i(A) < \sum_{i\in S} v_i(B)$. Let $C$ be an alternative that maximizes the welfare of $\bar{S}$. Then necessarily  $\sum_{i\in \bar{S}} v_i(C) \ge \sum_{i\in \bar{S}} v_i(A)$. By the fact that $S$ is a component, there must be an alternative $D$ for which $\sum_{i\in S} v_i(D) = \sum_{i\in S} v_i(A)$ and $\sum_{i\in \bar{S}} v_i(D) = \sum_{i\in \bar{S}} v_i(C)$. It follows that $\sum_{i\in {\cal{N}}} v_i(D) > \sum_{i\in {\cal{N}}} v_i(A)$, contradicting the assumption that $A$ maximizes welfare.
\end{proof}
}

A solution involves two aspects: a choice of alternative, and transfers.
In Proposition~\ref{lem:welfarecomponent} (Appendix \ref{app:decomp}) we show  that any alternative that maximizes welfare also maximizes welfare for each component separately.
%Proposition~\ref{lem:welfarecomponent}
Thus the proposition shows that every welfare maximizing solution respects the component structure of the given instance, as far as the choice of alternative is considered. For a solution to qualify as ``decomposable", it makes sense to in addition require that there are no transfers between components. We refer to this forbidding of transfers between components as {\em weak decomposability}.

\begin{definition}[weak decomposability]
\label{def:weakdecompose}
Let $\cal{N}$ be the set of agents, let $\cal{A}$ be the set of alternatives, and let $v$ be the tuple of valuation functions of the agents. A solution $(A,p)$, composed of an alternative $A \in {\cal{A}}$ (in this definition we do not require $A$ to be a  welfare maximizing alternative, as decomposability is relevant also to mechanisms that do not maximize welfare) and a vector $p$ of transfers (summing up to~0), is {\em weakly decomposable} if for every component $S \subset {\cal{N}}$ it holds that $\sum_{i\in S} p_i = 0$. Namely, the net transfer into the component is~0 (consequently, the same holds for the net transfer out of the component).
\end{definition}

As a trivial example, every solution that involves no transfers is weakly decomposable.

We also introduce a notion of {\em strong decomposability} that postulates that utilities of individual agents within a component are not influenced by decisions in other components. Unlike the notion of weak decomposability which is the property of a single solution, the notion of strong decomposability is a property of a mechanism and not just of a single solution.
In the context of our work in which we assume ``full information upon request", a {\em mechanism} $M$ is a mapping from instances to solutions. The input to $M$ is an instance $I$ of arbitrary size, composed of a set $\cal{N}$ of agents, a set $\cal{A}$ of alternatives, and a tuple $v$ of valuation functions of the agents. The output $M(I)$ is the proposed solution for the instance $I$, where the solution is composed of a winning alternative (in general, it is not required to be an alternative that maximizes welfare) and a vector of transfers. A mechanism can be randomized, in which case, given an input instance, the mechanism generates a distribution over solutions, and the proposed solution is a random sample from this distribution.

\begin{definition}[strong decomposability]
\label{def:strongdecompose}
We say that a mechanism $M$ is {\em strongly decomposable} if for every decomposable instance $I$, the output of the mechanism is consistent with the outputs of the mechanism on each of the components separately, in the following sense. Let $\cal{N}$ be the set of agents in $I$, let $\cal{A}$ be the set of alternatives, and let $v$ be the tuple of valuation functions of the agents.  Given ${\cal{N}}$, $\cal{A}$ and $v$, let $S\subset {\cal{N}}$ be a component. Let $I_S$ be the instance that results from restricting the set of agents of $I$ to be just $S$ (without changing the set of alternatives and the valuation functions of the agents in $S$). Let $M(I)$ ($M(I_S)$, respectively) denote the outcome (chosen alternative and vector of transfers) when $M$ is applied to instance $I$ ($I_S$, respectively). Then for every agent $i \in S$, her utility in both cases is the same. Namely, $u_i(M(I)) = u_i(M(I_S))$. (For randomized mechanisms, equality needs to hold for the expected utility.)
\end{definition}

Thus, strongly decomposable mechanisms essentially decide on the solution in each component independently of other components.
% Short version:
Proposition \ref{prop:stong-decomp-implies-weak} in Appendix \ref{app:decomp} shows that
strong decomposability implies weak decomposability in the following sense: assume that $M$ is a mechanism that for every instance selects an alternative that maximizes welfare and a budget balanced vector of transfers. If $M$ is strongly decomposable, then for every decomposable instance the solution produced by $M$ is weakly decomposable.
\longver{
\begin{proposition}
Let $M$ be a mechanism that for every instance selects an alternative that maximizes welfare and a budget balanced vector of transfers. If $M$ is strongly decomposable, then for every decomposable instance the solution produced by $M$ is weakly decomposable.
\end{proposition}

\begin{proof}
Let $S$ be a component. Let $A$ be the alternative (maximizing welfare for $\cal{N}$) chosen by $M$ in instance $I$, and let $A_S$ be the alternative (maximizing welfare for $S$) chosen by $M$ in instance $I_S$. By Proposition~\ref{lem:welfarecomponent} we have that $\sum_{i\in S} v_i(A) = \sum_{i\in S} v_i(A_S)$. By strong decomposability of $M$ we have that $u_i(M(I)) = u_i(M(I_S))$ for every $i\in S$, and consequently $\sum_{i\in S} u_i(M(I)) = \sum_{i\in S}u_i(M(I_S))$. As a utility of an agent is the sum of value for the chosen alternative and the transfer, we have that the sum of transfers of the agents in $S$ must be the same in $I$ and in $I_S$. But in $I_S$ the sum of transfers is~0, because of budget balance. Hence in $I$ the net transfer into $S$ is~0 as well, as required by weak decomposability.
\end{proof}
}

Observe that strong decomposability does not require that the same alternative be chosen in $I$ and in $I_S$, but rather only that agents in $S$ receive the same utilities in both instances (and likewise, as $\bar{S}$ is also a component, that agents in $\bar{S}$ receive the same utilities in $I$ and in $I_{\bar{S}}$). In Section~\ref{sec:decomp} we shall revisit decomposability in the context of the room allocation problem, and there we shall additionally require (and achieve) that the choices made by $M$ in different components give one single alternative for all of $I$.

We now present a proposition that describes the component structure of an instance.

Given a set $U$, a collection $\cal{C}$ of subsets of $U$ is a (distributive) {\em lattice} if for every two sets $S,T \in {\cal{C}}$ it holds that $S\cap T \in {\cal{C}}$ and $S\cup T \in {\cal{C}}$. The minimal sets of a lattice (a set from the lattice is minimal if it is nonempty and does not contain any other nonempty set from the lattice) form a partition of $U$. The following proposition is proved in Appendix~\ref{app:decomp}.

\begin{proposition}
\label{lem:component}
Given a set $\cal{A}$ of alternatives, a set $\cal{N}$ of agents, and a tuple $v$ of valuation functions, the components of $\cal{N}$ form a lattice.
\end{proposition}

It follows from Proposition~\ref{lem:component} that the minimal components form a partition of ${\cal{N}}$. (It could be that only ${\cal{N}}$ itself is a nonempty component.)  It can be shown by induction that that if each of the minimal components $S$ selects a most preferred alternative on its own, then there will be no conflicts between the choices -- we will be able to select a single alternative that is just as good, from the point of view of every player in every component.

{\bf The anticore and decomposition:}
A major benefit of the anticore is that it ensures decomposability properties. We remark that even though our notion of the anticore is the same as that of~\cite{Moulin92}, the notion of decomposability was not defined in that or other previous work, and hence the connection between anticore and decomposability is a new contribution of the current paper.  For weak decomposability (Definition~\ref{def:weakdecompose}) we have:

\begin{proposition}
\label{pro:upperdecomposes}
Every solution in the anticore is weakly decomposable.
\end{proposition}

\begin{proof}
Let $\cal{N}$ be the set of agents, let $\cal{A}$ be the set of alternatives, let $v$ be the tuple of valuation functions of the agents, and let $S \subset {\cal{N}}$ be a component. Consider an arbitrary solution $(A^*,p)$ in the anticore, composed of a welfare maximizing alternative $A^* \in {\cal{A}}$ and a vector $p$ of transfers. By Proposition~\ref{lem:welfarecomponent}, $A^*$ also maximizes the welfare of each of the components $S$ and $\bar{S}$ separately. By the anticore constraints, the net transfer into $S$ is at most~0, and so is the net transfer into $\bar{S}$. Consequently, the net transfer into $S$ is exactly~0. Hence the solution is weakly decomposable.
\end{proof}

It is premature at this stage to address strong decomposability (Definition~\ref{def:strongdecompose}). This will be done later, in Proposition~\ref{pro:strongdecompose}.

\section{Selection from within the welfare-sharing core}\label{sec:select}

The set of constraints corresponding to domination over a disagreement point are meant to achieve the property of having each agent (weakly) prefer (in terms of utility) every solution in the WS-core over the disagreement point. Our guideline for selecting a unique solution from within the WS-core (when  the WS-core is nonempty) is that we wish this property to hold not only in a qualitative manner, but also in a quantitative manner, to the largest extent possible. Ideally, we would like it to be that for every agent, switching to our mechanism offers a worthwhile increase in utility compared to the disagreement point. This calls for an {\em egalitarian} distribution of the welfare gain among all agents, where the welfare gain is the difference in welfare  between the maximum welfare alternative ($W_{max}({\cal{N}})$) and the expected welfare generated by the disagreement point ($W_{\pi}({\cal{N}})$). However, equal sharing of the welfare gain might not be in the WS-core, because it might violate the constraints of the anticore. Hence we aim to equalize the shares of the gain as much as possible, subject to satisfying the anticore constraints.

\subsection{Selection concepts}

Before proceeding, let us establish some conventions and notation.
We assume for convenience that the valuation function of each agent is such that at the disagreement point her expected value is~0. This can be enforced by applying an additive shift of $u_{\pi}(i)$ to each valuation function $v_i$. Given a solution $(A^*,p)$, we let $u_i$ denote the utility $u_i(A^*,p) = v_i(A^*) + p_i$ derived by agent $i$ from the solution (where the valuation function $v_i$ is such that the expected value offered by the disagreement point is~0). We shall sometimes refer to the vector $u = (u_1, \ldots, u_n)$ (rather than to $(A^*, p)$) as our solution, as this vector summarizes what the agents care about in a solution. An egalitarian solution will give every agent utility $u_i = \frac{W_{max}({\cal{N}})}{n}$, but might not be in the WS-core.  We present several different approaches for how to relax the egalitarian requirement so as to select a solution within the WS-core (when it is nonempty).

\begin{itemize}

\item The {\em min-square} solution. Here we seek the unique solution within the WS-core minimizing $\sum_{i \in {\cal{N}}} (u_i)^2$. This solution minimizes the variance in the distribution of the welfare, subject to being in the WS-core.

\item The {\em lexicographically-maximal} (\lexmaxws) solution. Given a vector $x \in R^n$, let $\hat{x}$ be the same vector with coordinates rearranged such that in the new order $\hat{x}_1 \le \hat{x}_2 \le \ldots \le \hat{x}_n$. For two vectors $x \in R^n$ and $y \in R^n$ of equal sum of their entries, $x \ge_{Lex} y$ denotes that for the rearranged vectors $\hat{x}$ and $\hat{y}$ and for some $1 \le k < n$ it holds that $\hat{x}_k > \hat{y}_k$, with $\hat{x}_i = \hat{y}_i$  for every $1 \le i < k$. A solution $u$ in the WS-core is {\em lexicographically maximal} if $u \ge_{Lex} u'$ for every other solution $u'$ in the WS-core.
%It can be shown that when the WS-core is not empty, the lexicographically maximal solution is unique (in terms of the utility that each agent gets).

\item A {\em Lorenz-maximal} solution. Given a vector $x \in R^n$, let $\hat{x}$ be the same vector with coordinates rearranged such that in the new order $\hat{x}_1 \le \hat{x}_2 \le \ldots \le \hat{x}_n$. For two vectors $x \in R^n$ and $y \in R^n$ of equal sum, we say that $x$ {\em Lorenz dominates} $y$ (denoted by $x \ge_{Lor} y$) if for the rearranged vectors $\hat{x}$ and $\hat{y}$ it holds that $\sum_{i=1}^k \hat{x}_i \ge \sum_{i=1}^k \hat{y}_i$, for every $1 \le k \le n$. A Lorenz maximal solution is a solution in the WS-core that Lorenz-dominates every other solution in the WS-core. By definition, it also minimizes the so called {\em Gini index} of inequality~\cite{Gini}. % [https://en.wikipedia.org/wiki/Gini_coefficient]

\end{itemize}

The next proposition presents some properties of the WS-core, for the proof see Appendix \ref{app:WS-core}.
\begin{proposition}
\label{pro:unique}
When the WS-core is nonempty:
\begin{enumerate}
\item
The min-square solution exists and is unique (in terms of the utility that it offers each agent).
\item The lexicographically-maximal solution exists and is unique.
\item The min-square solution and the lexicographically-maximal solution need not coincide.
    \item A Lorenz dominating solution need not exist.
     \item If a Lorenz dominating solution exists, it is unique, and moreover, it coincides both with the lexicographically-maximal solution and with the min-square solution.
         \end{enumerate}
\end{proposition}

Thus, the min-square and the lexicographically-maximal solutions exist whenever the WS-core is nonempty, whereas a Lorenz dominating solution need not exist.
Out of the two solutions that do exist, we suggest picking the {\em lexicographically-maximal-Welfare-Sharing} solution, which we denote by \lexmaxws.  (This choice is not of major significance to our work. Proposition~\ref{pro:water} (with a different algorithm) and Proposition~\ref{pro:strongdecompose} also hold with respect to the min-square solution, and Corollary~\ref{cor:submodular} shows that the two solutions coincide in many cases of interest.)

\subsection{The water filling algorithm}

When $W_{max}$ is submodular, the utilities in the \lexmax\ solution can be computed using an algorithm that we refer to as {\em water filling} (this is a generic name, used also elsewhere, for algorithms that increment variables at a uniform rate, subject to constraints). It proceeds in iterations. Initially (at iteration~0), all agents are {\em free} and every agent $i$ starts with her disagreement utility $u_{\pi}(i)$. If any of the constraints of the anticore are tight (satisfied with equality) by this initial solution, then the set $S_1$ of agents involved in the tight constraints become {\em locked}. Thereafter, in every iteration $j \ge 1$ we do the following. If there are no free agents, the algorithm ends and outputs the utilities of the agents.  If there are free agents, then the utility of every free agent is incremented by the same value $x_j$, where $x_j > 0$ is the smallest value that leads to some new anticore constraint becoming tight (equivalently, $x_j$ is the largest increase that does not violate any of the anticore constraints). At this point, the set $S_j$ of agents involved in a newly tight constraint become locked (some of these agents may have been locked already earlier), and iteration $j$ ends.

\begin{proposition}
\label{pro:water}
When $W_{max}$ is submodular, the water filling algorithm computes the \lexmax\ solution.
\end{proposition}

\begin{proof}
By design, the water filling algorithm satisfies all domination constraints and all constraints of the anticore. Likewise, it produces a lexicographically maximal solution subject to these constraints. It remains to show that by the end of the algorithm, the anticore constraint $\sum_{i\in {\cal{N}}} u(i) \le W_{max}({\cal{N}})$ is tight (meaning that all welfare has been shared by the agents), where here $u(i)=\sum_{t=1}^{j} x_j$ is the utility of agent $i$ that was locked right after iteration $j$ %at the end of the algorithm
(and $u(S)$ will serve as shorthand notation for $\sum_{i\in S} u_i$).

When the water filling algorithm ends, every agent is involved in a tight anticore constraint. Suppose for the sake of contradiction that the constraint $u({\cal{N}}) \le W_{max}({\cal{N}})$ is not tight. Let $S_j$ (for $j = 1, 2, \ldots$) be the sets of agents whose constraints became tight in iteration $j$ of
the execution of the water filling algorithm.
Consider a set function $W'_{max}$ that is identical to $W_{max}$, except that $W'_{max}({\cal{N}}) = u({\cal{N}}) < W_{max}({\cal{N}})$. Then $W'_{max}$ is submodular, all the sets $S_j$ referred to above are tight with respect to it, and so is $\cal{N}$. By Lemma~\ref{lem:lattice}, all unions of these sets $S_j$ are also tight with respect to $W'_{max}$. By repeatedly taking unions, we arrive at two sets $S$ and $T$ such that $S \cup T = {\cal{N}}$, both $u(S) = W'_{max}(S) = W_{max}(S)$ and $u(T) = W'_{max}(T) = W_{max}(T)$ hold, and $u(S \cap T) = W'_{max}(S\cap T) < W_{max}(S\cap T)$. By linearity, $u({\cal{N}}) = u(S) + u(T) - u(S \cap T) \ge W_{max}(S) + W_{max}(T) - W_{max}(S\cap T)$. By submodularity of $W_{max}$ we have that $W_{max}({\cal{N}}) \le W_{max}(S) + W_{max}(T) - W_{max}(S\cap T)$. These last two inequalities contradict our assumption that $u({\cal{N}}) < W_{max}({\cal{N}})$.
\end{proof}

%For the proof of the proposition see Appendix \ref{app:water}. We also provide there an example in which $W_{max}$ is not submodular and the water filling algorithm does not find the \lexmax\ solution, even though the WS-core is nonempty.

The water filling algorithm has at most $n$ iterations, because in every iteration at least one more agent becomes locked. For an agent $i$ that becomes locked at iteration $k$, her final utility is $u_{\pi}(i) + \sum_{j=1}^k x_j$. The total running time of the algorithm depends on the complexity of computing the disagreement utilities, and the computational cost of identifying a violated constraint in the anticore. We show (see Theorem~\ref{thm:complexity} for the exact statement) that when the disagreement utilities are given and $W_{max}$ is submodular, the \lexmax\ solution can be computed in time polynomial in $n$, but without the submodularity requirement, computing the \lexmax\ solution is NP-hard.

When $W_{max}$ is not submodular, the water filling algorithm might fail to find the \lexmax\ solution, even when the WS-core is nonempty. Consider an instance with three agents and three alternatives, with the following valuation functions.

\begin{tabular}{|c|c|c|c|}
  \hline
  % after \\: \hline or \cline{col1-col2} \cline{col3-col4} ...
   & {\bf Alternative 1} & {\bf Alternative 2} & {\bf Alternative 3} \\
\hline
  {\bf Agent 1} & 0 & -2 & 2 \\
  {\bf Agent 2} & 0 & 2 & -1 \\
  {\bf Agent 3} & 0 & 2 & -1 \\
  \hline
\end{tabular}

When Alternative~1 serves as the disagreement alternative, the WS-core is nonempty and contains a unique solution, given by choosing Alternative~2 (that has total welfare of~2), and transfers such that the utility of Agent~1 is~0 and of the other two agents is~1. Being unique, this is also the \lexmax\ solution.  However, the water filling algorithm will start by giving evert agent a utility of $\frac{1}{2}$. At this point every agent is part of a tight anticore constraint), and the algorithm ends before distributing all the welfare to the agents.

%We remark that the proof of Theorem~\ref{thm:DR89} contains a different iterative algorithm that also computes the \lexmax\ utilities.  The number of anticore constraints that need to be checked decreases there in every iteration, and hence that algorithm is more efficient than the one presented above. However, it applies only in the special case that $W_{max}$ is submodular (this is what causes the decrease in number of constraints), whereas the above algorithm template applies for every $W_{max}$.

\subsection{Decomposability of \lexmaxws}

The \lexmaxws\ solution lies in the anticore, and hence by Proposition~\ref{pro:upperdecomposes} it is weakly decomposable. We now consider strong decomposability (see Definition~~\ref{def:strongdecompose}). This property involves comparing the solutions generated for different instances. As the \lexmaxws\ solution for an instance depends on the disagreement point for the instance, we need to also relate between the disagreement points of different instances. For this purpose, we assume that there is a mechanism that given an instance outputs the disagreement point for that instance. For example, RP served as such a mechanism in Section~\ref{sec:background}.

\begin{proposition}
\label{pro:strongdecompose}
When $W_{max}$ is submodular, every mechanism $M$ that satisfies both following properties is strongly decomposable.
\begin{enumerate}
\item For every instance $M$ selects the respective \lexmaxws\ solution.
\item The disagreement utilities (that define the domination constraints for the WS-core) are the output of a disagreement mechanism that is strongly decomposable.
\end{enumerate}
\end{proposition}

\begin{proof}
The outcome of the water filling algorithm on the whole instance is identical to the concatenation of its outcomes on each component separately, because of Proposition~\ref{pro:upperconstraintsdecompose}.
\end{proof}

%Given Proposition \ref{pro:unique}, to show that all three solution concepts coincide  in that case that $W_{max}$ is submodular, it suffices to  show that a Lorenz dominating solution exists, and this indeed is what we prove in

\subsection{A Lorenz dominating solution}

We next show that in the important case that $W_{max}$ is submodular (e.g., in the \rentalprob), a Lorenz dominating solution necessarily exists.
%all three solution concepts are the same, all equal to our suggested solution \lexmaxws. See Theorem~\ref{thm:DR89} (in combination with item~5 of Proposition \ref{pro:unique}).
Theorem~\ref{thm:DR89}  below
is an adaptation of a theorem of Dutta and Ray~\cite{DR89}, which considers Lorenz minimal solutions and supermodular characteristic functions (we consider Lorenz maximal solutions and submodular characteristic functions). We provide a detailed proof of Theorem~\ref{thm:DR89} rather than attempt to use the results of~\cite{DR89} as a blackbox, because in our setting we need to ensure that the solution dominates a given disagreement point,  and this issue does not seem to have an analog in the setting of~\cite{DR89}.
%We remark that there are issues that are considered in the work of~\cite{DR89}, specifically -- the {\em Lorenz core}, that we need not concern ourselves with in our proof.
The proof of the following theorem appears in Appendix \ref{app:proof-thm-dr89}.

\begin{theorem}
\label{thm:DR89}
If $W_{max}$ is submodular, then the \lexmaxws\ solution (which is in the WS-core) % WS-core has a solution that
Lorenz-dominates all other  solutions in the WS-core.
\end{theorem}

From Proposition \ref{pro:unique} and Theorem \ref{thm:DR89} we derive the following corollary, implying that our solution \lexmaxws\ satisfies all three properties: it is a Lorenz dominating solution, the min-square solution and the lexicographically-maximal solution.

\begin{corollary}
\label{cor:submodular}
If the function $W_{max}$ is submodular then the WS-core is non-empty,
a Lorenz dominating solution exists, it is unique, and it coincides with both  the min-square solution and the lexicographically-maximal solution.
\end{corollary}

% \mbedit{We remark that in the submodular case, while Theorem \ref{thm:DR89} presents a polynomial-time algorithm for computing the \lexmaxws\ solution given an oracle access to $W_{max}$, even computing the welfare maximizing alternative is NP-hard \cite{}. However, for some structured settings, like in the \rentalprob\ discussed in Section \ref{sec:rental}, the welfare maximizing alternative can be computed in polynomial time, and the oracle can also be implemented by a computationally-efficient algorithm.      }
%\begin{proof}
%Follows from Theorem~\ref{thm:DR89} in combination with item~5 of Proposition~\ref{pro:unique}.
%\end{proof}

Summarizing, when selecting a solution from the WS-core, we employ the egalitarian paradigm. As we shall see, in several natural settings (such as the \rentalprob, % the Rental Harmony setting,
see Section~\ref{sec:rental}) $W_{max}$ is submodular. In these cases, Theorem~\ref{thm:DR89} offers a natural choice of a unique solution within the WS-core, because (essentially) all natural relaxations of the notion of being egalitarian (min-square, lexicographically-maximal, Lorenz-maximal) coincide. Moreover, in these cases the solution is computable in polynomial time, and also is continuous with a Lipshitz constant of~1 (see Appendices~\ref{sec:algorithms} and~\ref{sec:continuity} for exact statements.
%In those cases in which $W_{max}$ is not submodular, Proposition~\ref{pro:unique} shows that different relaxations of the egalitarian notion might lead to different solutions within the WS-core. We propose to select the lexicographically maximal solution \lexmaxws, though we view this particular proposal (as opposed to the min-square solution, for example) more as a matter of personal taste rather than as a central aspect of our solution concept.

%\subsection{Comparison with other solutions concepts}
\subsection{Comparison with the Shapley value}
\label{sec:shapley}

We find it particularly instructive to compare our \lexmaxws\ solution concept with that of the Shapley value~\cite{shapley1952}. In our context, it is natural to let $W_{max}$ play the role of a characteristic function in the definition of the Shapley value. Given a permutation $\sigma$ over the agents, let $S$ denote the set of agents that precede agent $i$ in $\sigma$. The {\em marginal contribution} of agent $i$ in $\sigma$ is $W_{max}(S \cup \{i\}) - W_{max}(S)$. The {\em Shapley value} of agent $i$ is her expected marginal contribution in a random permutation over the players. The sum of Shapley values of all players is exactly the maximum welfare $W_{max}({\cal{N}})$. The Shapley value mechanism selects an allocation that maximizes welfare, and arranges the transfers so that the utility of every player equals her Shapley value.

%A theoretical justification given for using the Shapley value is that it is the unique solution that satisfies three properties referred to as {\em symmetry}  (agents with the same valuation function receive the same utility), {\em zero player} (an agent whose valuation function is~0 for all alternatives does not receive nor pay any transfer) and {\em linearity}  (linear changes to $W_{max}$  lead to linear changes in the distribution of welfare).

%When the WS-core in nonempty, the solution offered by the Shapley value will in general be different than our \lexmaxws\ solution, and the reason for this is that \lexmaxws\ does not satisfy the linearity property with respect to $W_{max}$. Not satisfying properties associated with $W_{max}$ alone (linearity or other) is a natural consequences of the fact that our WS-core is not defined only by constraints derived from $W_{max}$, but also by constraints derived from the disagreement point. Hence even if $W_{max}$ remains unchanged and only the disagreement point changes, our solution will change, whereas the Shapley value solution would not change.

%It is known that in general cost sharing games, the Shapley value solution might be outside the cost sharing core. Likewise, in our more specialized setting (in which $W_{max}$ is not arbitrary, but rather derived from valuations over alternatives) the Shapley value solution might be outside the anticore. This is shown in~\cite{Moulin92}, and we also present an example demonstrating this in Appendix~\ref{app:Shapley-prop}.

When $W_{max}$ is submodular, the Shapley value resides in the anticore (this follows from the known fact that the Shapley value is in the cost-sharing core whenever the characteristic function is submodular). However, being oblivious to the reference point,  the Shapley value is sometimes not in the WS-core, even when $W_{max}$ is submodular. We shall see such examples in Appendix~\ref{app:Shapley-prop}.%Section~\ref{sec:rental}. %Hence known results imply that if $W_{max}$ is submodular the Shapley value is in the anticore  but do not extend to imply that the Shapley value is in the WS-core.
%In contrast, known results imply that Theorem~\ref{thm:BS} and Theorem~\ref{thm:DR89} hold for the anticore, and they do extend (as we show in the proofs provided) to the WS-core as well.

The use of the Shapley value as a solution concept for problems such as  \rentalprob\ was advocated in~\cite{Moulin92}. There is was shown that the Shapley value solution satisfies four properties that are referred to as {\em individual rationality}, {\em resource monotonicity}, {\em population monotonicity}, and {\em stand alone test}. Our \lexmaxws\ solution satisfies the stand alone test (which is equivalent to being in the anticore) and satisfies the domination property that is considerably stronger and more versatile than the notion of individual rationality used in~\cite{Moulin92}. However, it does not satisfy that resource  monotonicity and the population monotonicity properties, see Appendix~\ref{sec:monotonicity} for details.

%Given the above, it may appear that in those cases in which the disagreement point happens to be the same as the individual rationality condition used in~\cite{Moulin92}, the use of the Shapley value is preferable over \lexmaxws.

In Appendix~\ref{app:2-agents} we compare between the solutions offered by the Shapley value and the solutions offered by \lexmaxws,
% \mbcomment{what is WS? a typo? this is also the only place we say max-lex and not lex-max.  	Also: why do we next say ''implicitly'' ? maybe replace by "systematically" ? }, implicitly
covering all instances of two agents and two alternatives. In these instances the Shapley value solution does lie in the WS-core, and for some range of values it coincides with \lexmaxws.  For those ranges of values for which the two solutions differ, the \lexmaxws\ solution offers a more egalitarian sharing of welfare compared to that offered by the Shapley value.

{\bf Other solution concepts:}
Previous work proposes many other solution concepts for profit-sharing and cost-sharing games. See for example~\cite{Moulin85}, \cite{MoulinBook}, \cite{PS2007} and references therein, as well as Appendix~\ref{appendix:bargaining} in which we discuss the nucleolus and some known bargaining solution concepts. It would be too space consuming to discuss all these solution concepts in an informative manner, so let us just note that despite the many existing solution concepts, our solution concept appears to be new, and the first to satisfy the combination of properties that we seek. It is our opinion that in the setting studied in the current paper (which in particular includes a disagreement point), our solution concept is preferable over any of the previous solution concepts in terms of the balance that it achieves between fairness properties and the incentives that it provides to switch (from a mechanism with no transfers) to our mechanism.

\ignore{

}

%% file: the-appendix.tex
\appendix

\input{discussion}

\input{decomposability-app}

\section{Decomposability versus separability}\label{app:separable}

Recall the notion of decomposable solutions and mechanisms defined in Section~\ref{sec:decompose}. This notion appears to be new, but there are other known concepts that bear superficial similarity to decomposability. The purpose of this section is to present one such concept, that of separability, and clarify the differences between separability and decomposability.

Separability (see for example~\cite{MoulinBook} and references therein) of profit-sharing games is an internal consistency property for profit-sharing mechanisms. Suppose that that there is a profit-sharing mechanism $M$ that given a set $\cal{N}$ of players, disagreement utilities $u_i$ for each $i \in {\cal{N}}$, and a total welfare $W \ge \sum_{i\in {\cal{N}}} u_i$ to be divided among the players, assigns shares $x_i \ge u_i$ of the welfare to the respective agents, satisfying $\sum_{i\in {\cal{N}}} x_i = W$. Then $M$ is {\em separable} if it has the property that for every $S \subset {\cal{N}}$, had $M$ been applied to an instance in which $S$ is the set of players, in which the same $u_i$ as above (for $i \in S$) are their disagreement utilities, and in which the total welfare to share is $\sum_{i\in S} x_i$, then $M$ would propose the same shares $x_i$ as above (for $i \in S$). A simple example of a separable mechanisms is the egalitarian mechanism that equalizes the profit (above disagreement point) of all agents.

A major difference between decomposability and separability is that decomposability (of a solution, or a mechanism) is a property that needs to hold only for decomposable instances, whereas separability needs to hold for all instances. Consequently, decomposable mechanisms need not be separable, and in fact our \lexmax\ mechanism is not separable. If $S\subset {\cal{N}}$ is not a component, then the way the part of the welfare allocated to a subset $S$ of players is distributed among the players of $S$ is affected by the way the remaining welfare is distributed among members of ${\cal{N}} \setminus S$, because the anticore constraints involve constraints on sets $T$ of players that intersect both $S$ and ${\cal{N}} \setminus S$, and these constraints affect which share allocations within $S$ are feasible. Likewise, separable mechanisms need not be decomposable (and in general are not). For example, the egalitarian profit-sharing mechanism (equalizing $x_i - u_i$) is separable, but fails to satisfy even the weak decomposability property.

A property that is related to separability and can take into account the anticore constraints is {\em consistency for reduced games}. This property is satisfied by a solution concept referred to as the Nucleolus. This property can coexist with decomposability and can be incorporated into our solution concept by changing the rule that selects a solution from the WS-core, but we prefer not to enforce this property for reasons explained in Section~\ref{app:nucleolus}.

% See Section~\ref{app:nucleolus} for a discussion.

\section{More on the anticore}\label{app:anticore}

The following  property of the anticore plays a central role in decomposability aspects.

\begin{proposition}
\label{pro:upperconstraintsdecompose}
For a given decomposable instance, let $T_1, \ldots T_t$ be a partition of the set $\cal{N}$ of agents into components (in the sense of Definition~\ref{def:decompose}). Then a solution $(A,p)$, composed of a chosen alternative $A$ and vector $p$ of transfers, satisfies the anticore constraints ($u_S(A,p) \le W_{max}(S)$ for every set $S \subseteq{\cal{N}}$) if and only if it satisfies them for every set $S$ that is fully contained in a component (namely, $S \subseteq T_j$ for some $1 \le j \le t$).
\end{proposition}

\begin{proof}
The ``only if' direction is a triviality. For the ``if" direction, assume for the sake of contradiction that there is a set $S$ such that $u_S(A,p) > W_{max}(S)$, but for every $1 \le j \le t$ it holds that $u_{S\cap T_j}(A,p) \le W_{max}(S\cap T_j)$. Then for every $j$ there is an alternative $A_j$ such that $u_{S\cap T_j}(A,p) \le u_{S \cap T_j}(A_j)$, and then Definition~\ref{def:decompose} implies that there is some alternative $A*$ such that $u_{S\cap T_j}(A,p) \le u_{S \cap T_j}(A^*)$ simultaneously for all $j$. Consequently, $u_S(A,p) \le u_S(A^*) \le W_{max}(S)$, which contradicts our assumption.
\end{proof}

\section{Properties of the WS-core}

\input{main-thm-proof.tex}

\subsection{Proof of Proposition~\ref{pro:unique}}
\label{app:WS-core}

We next restate and prove Proposition~\ref{pro:unique}.

\begin{proposition}
%	\label{pro:unique}
	When the WS-core is nonempty:
	\begin{enumerate}
		\item
		The min-square solution exists and is unique (in terms of the utility that it offers each agent).
		\item The lexicographically-maximal solution exists and is unique.
		\item The min-square solution and the lexicographically-maximal solution need not coincide.
		\item A Lorenz dominating solution need not exist.
		\item If a Lorenz dominating solution exists, it is unique, and moreover, it coincides both with the lexicographically-maximal solution and with the min-square solution.
	\end{enumerate}
\end{proposition}

\begin{proof}
	
	\begin{enumerate}
		\item
		The WS-core is a bounded polytope (with constraints $u \ge 0$, $\sum_{i} u_i = W_{max}({\cal{N}})$, and the anticore constraints). The function $\sum_{i \in {\cal{N}}} (u_i)^2$ is strictly convex. Hence it has a unique minimum in the WS-core.
		\item
		Assume for the sake of contradiction that there are two lexicographically maximal solutions $u$ and $u'$ with $u \not= u'$. %Then $\hat{u} = \hat{u'}$.
As the WS-core is a convex set, it holds that also $u" = \frac{u + u'}{2}$ is in the WS-core. It is not difficult to see that either $u" >_{Lex} u$ or $u" >_{Lex} u'$ (with a strict inequality). This contradicts the assumptions that both $u$ and $u'$ are lexicographically maximal.
		
		\item
		To see that the min-square solution and the lexicographically-maximal solution need not coincide, consider the following example with four alternatives and six players:
		\medskip

		\begin{tabular}{ccccccc}
			Example2 & Agent 1 & Agent 2 & Agent 3 & Agent 4 & Agent 5 & Agent 6 \\
			Alternative A & 1 & 1 & 1 & 1 & 3 & 3 \\
			Alternative B & 0 & 2 & 2 & 2 & 2 & 2 \\
			Alternative C & -1 & -1 & -1 & -1 & 4 & -9 \\
			Alternative D & 0 & -2 & -2 & -2 & -9 & 4
		\end{tabular}
		\medskip
		
		The disagreement distribution $\pi$ is uniform over the four alternatives. Under this distribution, the expected utility of each agent is~0. (In the above example $\pi$ is supported on Pareto-optimal alternatives. If this is not required, we may replace Alternatives C and D by a single disagreement alternative that gives utility~0 for every agent, thus simplifying the example.) Any of the first two alternatives maximizes welfare, giving a welfare of~10. An egalitarian distribution of the welfare will give each agent utility 5/3, but this is not in the anticore because it violates the constraint for Agent~1 (whose utility is upper bounded by~1). The lexicographically maximal solution is to choose Alternative A and have no transfers. This gives Agent 1 his maximum possible value of~1, and then conditioned on that, it gives each of Agents~2, 3 and~4 their maximum utility of~1 (because agent~1 and agent~$j \in \{2,3,4\}$ combined are not allowed to have utility above~2). However, this is not a min-square solution: choosing Alternative B with no transfers gives a lower value for the sum of squares (20 instead of 22).
		
		\item In the above example, there is no Lorenz dominating solution, because the lexicographically maximal solution (utilities as in Alternative A) does not Lorenz dominate the solution corresponding to the utilities under Alternative B.
		
		\item If follows essentially by definition that $x \ge_{Lor} y$ implies that also $x \ge_{Lex} y$. Hence Lorenz domination implies being lexicographically maximal. To obtain the Lorenz dominating vector of utilities from any other vector of utilities (that maximizes welfare and is in the WS-core), one needs to shift utility from coordinates of higher utility to coordinates of lower utility, by this lowering the sum of the squares of the utilities. This shows that a Lorenz dominating solution is also a min-square solution.
	\end{enumerate}
\end{proof}

%\section{The \lexmax\ solution}

%\input{water-filling-proof}
\input{Lorenz-domination-proof}

\input{sec-rent}

\section{Properties of the Shapley-value solution}\label{app:Shapley-prop}

A theoretical justification given for using the Shapley value is that it is the unique solution that satisfies three properties referred to as {\em symmetry}  (agents with the same valuation function receive the same utility), {\em zero player} (an agent whose valuation function is~0 for all alternatives does not receive nor pay any transfer) and {\em linearity}  (linear changes to $W_{max}$  lead to linear changes in the distribution of welfare).

When the WS-core in nonempty, the solution offered by the Shapley value will in general be different than our \lexmaxws\ solution, and the reason for this is that \lexmaxws\ does not satisfy the linearity property with respect to $W_{max}$. Not satisfying properties associated with $W_{max}$ alone (linearity or other) is a natural consequences of the fact that our WS-core is not defined only by constraints derived from $W_{max}$, but also by constraints derived from the disagreement point. Hence even if $W_{max}$ remains unchanged and only the disagreement point changes, our solution will change, whereas the Shapley value solution would not change.

It is known that in general cost-sharing games, the Shapley value solution might be outside the cost-sharing core. Likewise, in our more specialized setting (in which $W_{max}$ is not arbitrary, but rather derived from valuations over alternatives) the Shapley value solution might be outside the anticore. This is shown in~\cite{Moulin92}. %, and we also present an example demonstrating this in Appendix~\ref{app:Shapley-prop}.
For completeness, we also provide such an example. % in which the Shapley value solution is outside the anticore.
Consider three agents and four alternatives, with valuation functions as in the following table:

\begin{tabular}{|c|c|c|c|c|}
  \hline
  % after \\: \hline or \cline{col1-col2} \cline{col3-col4} ...
  Example 3 & {\bf Alternative 1} & {\bf Alternative 2} & {\bf Alternative 3} & {\bf Alternative 4} \\
\hline
  {\bf Agent 1} & 2 & 0 & 0 & 1 \\
  {\bf Agent 2} & 0 & 2 & 0 & 1 \\
  {\bf Agent 3} & 0 & 0 & 2 & 2 \\
  \hline
\end{tabular}

Alternative~4 generates the highest welfare.
The anticore restricts the combined utility of the first two agents to at most~2, and likewise for the utility of Agent~3. Hence in the anticore the utility of agent~3 is exactly~2. However, the Shapley value for agents~1 and~2 is 7/6, and for agent~3 it is 5/3. Hence in the Shapley value mechanism agent~3 pays agents~1 and~2.

In the above example $W_{max}$ is not submodular: the marginal contribution of agent~1 to the set $\{2,3\}$ is~1, whereas her marginal contribution to the set $\{2\}$ is~0.

\subsection{Proof of Proposition~\ref{prop:Shapley-prop}}

We restate and prove Proposition~\ref{prop:Shapley-prop}.

\begin{proposition}\label{prop:Shapley-prop-app}
	The Shapley-value solution is in the anticore of the \rentalprob\ and satisfies strong decomposability, yet it does not dominate \emph{RP}.
\end{proposition}

\begin{proof}
	As noted earlier, the function $W_{max}$ is submodular for the \rentalprob, and the Shapley value for submodular cost-sharing games lies in the cost-sharing core. As mathematically the anticore has the same definition as the cost-sharing core, it follows that the Shapley value is in the anticore.
	
	To see that the Shapley value is strongly decomposable, recall that the Shapley value of an agent $i$ is computed by considering a uniform distribution over all permutations, and taking the expected marginal value (with respect to $W_{max}$)  of the agent over a random choice of permutation.  As the instance is decomposable, these marginals only depend on those agents from the part $P_j$ that contains $i$ that arrived before agent $i$. The uniform distribution for permutations over all agents indices a uniform distribution over the permutations for the agents in part $P_j$. Hence the vector of Shapley values for all agents  is simply the concatenation of the vectors of Shapley values for each of the parts, implying strong decomposability.

	The fact that it does not dominate \emph{RP}
	% This
	follows from the fact that the Shapley value mechanism is continuous, whereas the \emph{RP} mechanism is not, not even when its output maximizes welfare (see Proposition  \ref{prop:rp-non-cont}). Consequently, computing the utility of the agents under the Shapley value mechanism for the example used in the proof of
	% 	item~2 of Proposition~\ref{pro:continuousANT}
	Proposition \ref{prop:rp-non-cont}
	will provide an example proving the current proposition.
\end{proof}

\section{Comparison to some prior solution concepts}
\label{appendix:bargaining}

Many bargaining solutions have been suggested in the past. Below we discuss some of the most known ones, and show that none of them satisfy all the properties we consider desirable.
The solutions we discuss include the Nash bargaining solution and the Kalai-Smorodinsky
bargaining solution \cite{KalaiS75}. %, the Shapley value and  the Nucleolus.
As transfers are allowed, each of them picks a division of the utility generated by the social maximizing outcome among the players. These solutions can be defined to dominate RP, but need not lie in the anticore. This last fact has the following undesirable consequences.

\begin{proposition}
The Nash bargaining solution is not reasonable from above ($u_i \le W_{max}(i)$ might fail for some agent $i$), and
it fails to satisfy even the weak decomposition property.
\end{proposition}

\begin{proposition}
\label{pro:KS}
The Kalai-Smorodinsky bargaining solution is  reasonable from above, % ($u_i \le W_{max}(i)$ for each agent $i$),
but it fails to satisfy even the weak decomposition property.
\end{proposition}

Below we present examples proving the above propositions. We also discuss another solution concept, the {\em nucleolus}. We explain how a version of this notion can be defined so that it will be in our WS-core. As such, it can potentially serve as an alternative to the solution that we propose to select from the WS-core, namely, the \lexmax\ solution. Faced with these two alternatives, we explain why we prefer \lexmax\ over the version of the nucleolus that resides in the WS-core.
%TO DO:

%Nash requires Independence of irrelevant alternatives (IIR), while KS does not.

\subsection{The Nash bargaining solution}

The Nash bargaining solution, when applied in our setting, becomes identical to the egalitarian solution, equalizing the gains of all agents above their respective disagreement utilities.   As such, it dominates RP, but need not lie it the anticore. Moreover, the Nash bargaining solution need not be reasonable from above.  Namely, it may hold that $u_i > W_{max}(i)$ for some agent $i$. The following example illustrates this point.

Consider a setting with two agents and a two alternatives. The first agent has values of $24$ and $0$ for the two alternatives, while the second agents has value of $0$ and $4$.   Random priority will give the agents expected utilities of $12$ and $2$, respectively.  These are the disagreement utilities, and their total is only $14$, which is $10$ less than the maximum welfare of $24$.  The Nash bargaining solution will give each of them an additional utility of $5$ so the final utilities will be $17$ and $7$.  In particular, the utility of the second agent is higher than the utility offered to her by the best alternative.

Likewise, the Nash bargaining solution fails to satisfy the even the weak decomposition property. The example in Appendix \ref{sec:KS} can serve to illustrate this last point (details omitted).

% Consider the example in Section \ref{sec:KS}. The Nash Bargaining solution will equalize the amount each agent gets from the extra amount of 8 that needs to be divided among the agents beyond the disagreement utilities. Thus, each will get additional utility of $2$, and the final utilities will be $8,10,8$ and $14$. The first two agents together get $18$ instead of $16$ (the maximal they could get in any alternative), and the last two agents get $22$ instead of $24$. This means that there is a transfer between the two components, not satisfying the weak decomposition property (or the group upper bound for the first two agents).

\subsection{The Kalai-Smorodinsky bargaining solution}
\label{sec:KS}

The Kalai-Smorodinsky (KS) solution has the following geometric interpretation.
Each alternative corresponds to a point in $\Re_+^n$, specifying the utility of each of the  $n$ agents.
The \emph{point of best utilities} is the point with coordinate for each agent equal to the best utility she gets in any feasible alternative.
The Kalai-Smorodinsky solution is the intersection of the Pareto frontier with the line from the disagreement point to the point of best utilities.

We note that while our solution aims to equalize utility gains (compared to disagreement utilities) as much as possible (subject to anticore constraints), the KS solution aims to equalize the fraction of the utility improvements (from the disagreement to the best point).

Like in our solution, the KS solution never gives any agent more than her utility in the best alternative for her.
Yet, the KS solution might fail to satisfy the anticore constraints for groups of more than one player.
Consequently, it fails to satisfy even the weak decomposition property. This is illustrated by the following example:

Consider an item allocation setting in which there are four agents and four items, and in each alternative each agent receives one item. Hence there are $4! = 24$ possible alternatives. The valuation of each agent for each item is presented  in the following table:

\begin{tabular}{|c|c|c|c|c|}
	\hline
	% after \\: \hline or \cline{col1-col2} \cline{col3-col4} ...
	Example & {\bf Item 1} & {\bf Item 2} & {\bf Item 3} & {\bf Item 4} \\
	\hline
	{\bf Agent 1} & 4 & 8 & 0 & 0 \\
	{\bf Agent 2} & 4 & 12 & 0 & 0 \\
	{\bf Agent 3} & 0 & 0 & 4 & 8 \\
	{\bf Agent 4} & 0 & 0 & 4 & 20 \\
	\hline
\end{tabular}

The problem decomposes to the first two agents (with maximum welfare $16$) and last two agents (with maximum welfare $24$), giving a total maximum welfare $40$. Random priority will give the agents expected utilities of $6, 8, 6$ and $12$, respectively.
These are the disagreement utilities, and their total is only $32$, which is $8$ less than the maximum welfare.

The best utility each of the agents can get is $8, 12, 8$ and  $20$, respectively, for a total of $48$.
Subtracting the disagreement values this gives the agents $2, 4, 2$ and  $8$, respectively, which in total is $16$ above the disagreement point,
rather than $8$ that can be divided among the agents.

The KS solution lowers what each agent gets (above disagreement) in a uniform rate, meaning that the agents will get utilities $7, 10, 7$ and $16$, respectively. That is, each agent gets half the gap between his disagreement utility and his best utility.
Hence the first two agents together get $17$ instead of $16$ (the maximal they could get in any alternative), and the last two agents get $23$ instead of $24$. This means that there is a transfer between the two components, not satisfying the weak decomposition property (or the anticore constraint for the first two agents).

\subsection{The Nucleolus}
\label{app:nucleolus}

The {\em nucleolus}~\cite{Schmeidler69} is a solution concept that gives a unique solution to a coalitional game. If the core of the game is nonempty, the nucleolus resides in the core. Part of its attractiveness is due to possessing a property referred to as consistency for reduced games (see~\cite{MoulinBook} or~\cite{MPT92} for more details).

To apply this solution concept, one needs to describe our setting as a coalitional game, where there is a characteristic function associated with sets of agents. The question becomes which characteristic function to use. In general, readers may propose whatever characteristic function they find appropriate, and at this level of generality we have nothing to say about the nucleolus. However, part of our work concerns exactly the issue of selecting characteristic functions that we find appropriate for our setting, and they can serve as the basis for the definition of the nucleolus.

Our notion of anticore and the fact that we have a profit game (in which agents share the welfare increase that results from selecting a maximum welfare alternative) suggests the use of $D(S) = W_{max}({\cal{N}}) - W_{max}({\cal{N}} \setminus S)$ as a characteristic function (see the end of Section~\ref{sec:uppercore}). Applying the nucleolus framework to this characteristic function is then equivalent to selecting a solution that satisfies the anticore constraints with as large a margin as possible (measured in terms of a lexicographic vector of the margins). This solution satisfies the weak decomposability property (because it is in the anticore, see Proposition~\ref{pro:upperdecomposes}), but need not satisfy the domination constraints (and hence might not be in the WS-core). Even in the room (item) allocation setting (which is submodular), there are examples in which the nucleolus fails to dominate the Random Priority mechanism. The instance described in Appendix~\ref{app:RP} serves as one such example (for the same reason that it shows that the Shapley value solution does not dominate RP -- see the last paragraph of the proof of Proposition~\ref{prop:Shapley-prop-app} for an explanation).

To force the nucleolus solution to lie within the WS-core, one may modify the characteristic function so that its value on each set $S$ is the maximum between $D(S)$ and the sum of the disagreement utilities of agents in $S$. We refer to this function as the WS characteristic function, and to the resulting nucleolus concept as {\em nucleolus-WS}.  It selects a solution that satisfies all constraints of the WS characteristic function with as large a margin as possible (measured in terms of a lexicographic vector of the margins).
In contrast, our \lexmaxws\ solution satisfies only the domination constraints with as large a margin as possible, subject to not violating the anticore constraints. Having a large margin for the domination constraints serves the purpose of giving agents as strong as possible incentives to give up the reference point and adapt our mechanism.
Keeping a large margin from all constraints (including the anticore constraints) serves a different purpose: it attempts to keep the selected solution away from all boundaries of the feasible region. We prefer the \lexmaxws\ solution over nucleolus-WS, though using nucleolus-WS is also a reasonable option for selecting a solution from the WS-core. In particular, like the \lexmaxws, the nucleolus-WS mechanism satisfies the strong decomposability property. This is a consequence of being in the anticore (that ensures weak decomposability), together with the fact that the nucleolus satisfies consistency for reduced games (details omitted).

Section~\ref{app:2-agents} presents (among other things) a comparison between \lexmaxws\ and nucleolus-WS when there are two agents.

\section{Is $W_{max}$ submodular for submodular Combinatorial Valuations?}
\label{app:other-comb}
Proposition \ref{prop:unit-demand-ours} is the result of the fact that for unit-demand valuations the set function $W_{max}$ is submodular.
% We observe that the only property that we really used about the fact that the problem is unit-demand in the proof of Proposition \ref{prop:unit-demand-ours} is the fact that when agents are unit-demand the set function $W_{max}$ is submodular.
Thus, a similar claim is true for other classes of combinatorial valuation functions over the items that ensure that the  set function $W_{max}$ is submodular. One would be tempted to think that the fact that unit-demand  valuations are submodular was sufficient to prove that $W_{max}$ is submodular.
Yet care should be taken.
The submodularity of unit-demand valuations is for each valuation function, with respect to sets of \emph{items}, while the submodularity of $W_{max}$ is for the total welfare, and with respect  to sets of \emph{agents}.  Indeed, we next show that the fact that the class of submodular valuation functions does \emph{not} ensure that the  set function $W_{max}$ is submodular.

\begin{proposition}\label{prop:EF-paid}
	There are instances of allocation problems in which all agents have valuation functions that are submodular (and in fact, also belonging to the classes of {\em budget additive} valuations and {\em coverage} valuations) but nevertheless, the associated set function $W_{max}$ is not submodular.
\end{proposition}

\ignore{We next restate and prove Proposition \ref{prop:EF-paid}.

\begin{proposition}%\label{prop:EF-paid}
	There are instances of allocation problems in which all agents have valuation functions that are submodular (and in fact, also belonging to the classes of {\em budget balanced} valuations and {\em coverage} valuations) but nevertheless, the associated set function $W_{max}$ is not submodular.
\end{proposition}
}

\begin{proof}
	Consider an example with four agents and three items. The values that agents associate with individual items are presented in the following table.
	
	\begin{tabular}{|c|c|c|c|c|}
		\hline
		% after \\: \hline or \cline{col1-col2} \cline{col3-col4} ...
		Example 4 & {\bf Item 1} & {\bf Item 2} & {\bf Item 3} \\
		\hline
		{\bf Agent A} & 1 & 0 & 0  \\
		{\bf Agent B} & 0 & 2 & 0  \\
		{\bf Agent C} & 0 & 0 & 2  \\
		{\bf Agent D} & 2 & 1 & 1  \\
		\hline
	\end{tabular}
	
	Given the values of individual items, the valuation functions of each of the agents $A$, $B$ and $C$ are {\em additive}, and the valuation function of agent $D$ is {\em budget additive} with a budget of~2 (it also happens to be a {\em coverage function}).  Hence any set of either two or three items has value as~2 for agent $D$, even if the sum of item values is larger than~2. All four valuation functions are submodular. However, the set function $W_{max}$ is not submodular. Consider the sets $S = \{A,B,D \}$ and $T = \{A, C, D \}$. Then $4 + 4 = W_{max}(S) + W_{max}(T) <  W_{max}(S \cap T) + W_{max}(S \cup T) = 3 + 6$, violating the submodularity inequality.
\end{proof}

% The proof of the proposition appears in Appendix \ref{app:EF-prop}.
We conclude by remarking that if all agents have additive valuation functions then is is true that the set function $W_{max}$ is submodular, and Proposition \ref{prop:unit-demand-ours} applies to this setting as well.

%\section{The \emph{Random Priority} mechanism in not continuous}
\section{Both \emph{Random Priority} and Eating are not continuous}
\label{app:RP}

\begin{proposition}\label{prop:rp-non-cont}
	Neither the \emph{Random Priority (RP)} mechanism nor the \emph{Eating} mechanism are continuous, not even on instances on which they nearly maximize welfare.
\end{proposition}
\begin{proof}
	Consider an instance with three players and three items. For small $\epsilon > 0$, the valuation vectors of the players are $(1, 1 - \epsilon, \epsilon)$, $(1, 1 - \epsilon, \epsilon)$, $(1, 0, \epsilon)$. The maximum welfare allocation has welfare 2, and so does each allocation that might result from $RP$ and $Eat$ (because it is never the case that player~3 gets item~2). The expected utilities under $RP$ are roughly $(\frac{5}{6}, \frac{5}{6}, \frac{1}{3})$ (up to $O(\epsilon)$), and likewise for $Eat$ (in both cases player~3 has probability $\frac{1}{3}$ of getting item~1). Consider now a slightly modified set of valuation vectors: $(1-\epsilon, 1, \epsilon)$, $(1-\epsilon, 1, \epsilon)$, $(1 - \epsilon, 0, \epsilon)$. Again, the maximum welfare allocation has welfare 2, and so does each allocation that might result from $RP$ and $Eat$ (because it is never the case that player~3 gets item~2). The expected utilities under $RP$ are roughly $(\frac{2}{3}, \frac{2}{3}, \frac{2}{3})$ (up to $O(\epsilon)$), and likewise for $Eat$ (in both cases player~3 has probability $\frac{2}{3}$ of getting item~1). As $\epsilon$ tends to~0, every player suffers discontinuity in its utility.	
\end{proof}

\section{Properties of envy-free mechanisms}
\label{app:EF-prop}

%Proposition \ref{prop:EF-main-problems} easily follows from the next proposition.

\begin{proposition}\label{prop:EF-problems-app}
	There are instances of the unit-demand matching problem in which in each instance there is unique envy-free solution, and this solution fails to satisfy some desirable properties, as follows.
	%In each case below we present an instance of the rental harmony problem with a unique envy-free solution that:
	
	\begin{enumerate}
		\item It does not satisfy decomposability.
		
		\item It does not dominate the \emph{Random Priority} allocation, and moreover, some agent $i$ exceeds his own upper bound on utility $W_{max}(i) = max_j v_i(j)$.
		
		\item It does dominate the \emph{Random Priority} allocation and even the \emph{Eating} mechanism, but nevertheless, some  agent $i$ exceeds his own upper bound on utility $W_{max}(i)$.
		
		\item It is in the anticore but still fails to dominate \emph{Random Priority}.
		
		\item It dominated \emph{Random Priority} (and even \emph{Eating})  and no agent $i$ exceeds his own upper bound on utility $W_{max}(i)$, but is not in the anticore.
	\end{enumerate}
\end{proposition}

\begin{proof}
	We write valuation functions as vectors of item values. In all cases an allocation that maximizes welfare is the identity permutation. (It can be made unique by adding a small $\epsilon > 0$ to every $v_i(i)$, but this is omitted from the examples so as to keep the notation simple.)
	\begin{enumerate}
		\item See example in the introduction.
		
		\item Valuations (6,0,0), (6,0,0), (0,6,6). The only envy free transfer is $(-4,2,2)$. Agents~1 and~2 get less utility than in \emph{Random Priority} ($2$ instead of $3$), and agent~3 gets more utility than his highest value ($8$ instead of $6$).
		
		\item Valuations (6,0,0), (6,0,0), (1,0,0). The unique envy free transfer is $(-4,2,2)$. In both \emph{Random Priority} and \emph{Eating} each agent get the first item with probability $1/3$. Agent~3 gets more utility than his highest value ($2$ instead of $1$).
		
		\item Valuations (2,1,0), (2,1,0), (0,1,0). The unique envy free transfer is (-1,0,1). This satisfies the anticore constraints, but does not dominate \emph{Random Priority}  (each of agents~1 and~2 gets utility~1, whereas in {\em Random Priority} they get expected utility~$\frac{7}{6}$).
		
		\item Valuations  (7,0,0,0), (7,0,0,0), (2, 0, 1, 0), (2, 0, 1, 0).  The unique envy free transfer is $(-5,2,1,2)$. This dominates the \emph{Eating} mechanism and no agent $i$ exceeds his own upper bound on utility $W_{max}(i)$, but it violates the core upper bound constraints on some group (the group $\{3,4\}$ has utility $4$, higher than their upper bound of $3$).
	\end{enumerate}
\end{proof}

\subsection{Envy-free solutions might pay a renter that got his top room}
\label{sec-PRO:RENT}

We restate and prove Proposition \ref{PRO:RENT}.
\begin{proposition}
	There are instances of the \rentalprob\ % rental harmony
	in which the valuation function of each player is nonnegative and sums up to the total rent (this is the setting studied in~\cite{GMPZ}), but nevertheless there is a player that in every envy-free solution both gets his most desired room and receives more money than his rent share.
\end{proposition}

\begin{proof}
	Consider an instance with five players (and five rooms) in which they need to pay rent of $1$.
	For $0 < \epsilon < \frac{1}{8}$, the valuations functions are (in vector notation):
	
	$$v_1 = v_3 = (1 - \epsilon, 0, \epsilon, 0, 0)$$
	
	$$v_2 = v_4 = (0, 1 - \epsilon, 0, \epsilon, 0)$$
	
	$$v_5 = (0, 0, \frac{1 - \epsilon}{3},  \frac{1 - \epsilon}{3}, \frac{1+2\epsilon}{3})$$
	
	Observe that these valuation functions are normalized (each sums up to~1) and nonnegative.
	
	Every maximum welfare solution will place players~1 and~3 in rooms~1 and~3, players~2 and~4 in rooms~2 and~4, and player~5 in room~5. In particular, player~5 receives the room that he desires most. In an envy free solution, how much should each player pay?  To avoid envy between players~1 and~3, room~1 needs to be priced exactly $1 - 2\epsilon$ above room~3. Likewise, room~2 needs to be priced exactly $1 - 2\epsilon$ above room~4. To avoid envy for agent~5, room~5 needs to be priced at most~$\epsilon$ above any of rooms~3 and~4. Using the above information, the maximum that player~5 can be charged is at most $\frac{-1 + 8\epsilon}{5}$ (with players~1 and~2 getting rooms~1 and~2 and charged $\frac{4 - 7\epsilon}{5}$ each, and players~3 and~4 getting rooms~3 and~4 and charged $\frac{-1 + 3\epsilon}{5}$ each). For $0 < \epsilon < \frac{1}{8}$ players~1 and~2 more than cover the rent, and each of players~3, 4 and~5 receives money rather than pays rent. In a sense, players~1 and~2 are paying the other players so that the other players agree to give them rooms~1 and~2. It is perhaps justified that players~3 and~4 get paid -- they really wanted rooms~1 and~2, and instead got rooms of almost no value. Hence they sacrificed something to enable the solution, and may deserve some compensation. In contrast, player~5 got his most desirable room and sacrificed nothing, but nevertheless, is also getting paid.
\end{proof}

\section{A complete analysis for unit-demand matching with two agents}\label{app:2-agents}

In this section we present a complete analysis of unit demand matching when $n=2$ (two agents and two items). This is fact captures all settings in with two agents and two \emph{alternatives} (here the two alternatives are the two possible permutations over the items). We compare between the following four mechanisms:

\begin{enumerate}

\item Envy-free. To select a unique solution among all envy-free solutions, we use max-min as the selection rule, as was recommended in~\cite{GMPZ}.

\item The Shapley value solution.

\item Kalai-Smorodinsky bargaining. Here we use RP (or equivalently for the case $n=2$, {\em Eating}) as the disagreement mechanism.

\item Our \lexmaxws, again with RP as the disagreement mechanism.

\end{enumerate}

Let us denote the agents by $\{A,B\}$ and the items by $\{a,b\}$. After some normalization, we may assume that the valuation function for $A$ is $v_A(a) = 1$ and $v_A(b) = -1$, and for $B$ is $v_B(a) = \delta$ and $v_B(b) = -\delta$, with $-1 \le \delta \le 1$. The maximum welfare allocation (with a ties in case that $\delta = 1$) allocates item $a$ to agent $A$ and item $b$ to agent $B$. The following table presents the pair of utilities (for agent $A$ and $B$) under
the mechanisms {\bf max-min-EF} (maxmin envy free), {\bf Shapley value}, {\bf KS} (Kalai-Smorodinsky bargaining) and our  {\bf \lexmaxws}.
%the mechanisms $max-min-EF$ (maxmin envy free), $SV$ (Shapley value) and $MSM$ (min-square  mechanism).

\begin{tabular}{c c c c c}
	\hline
	% after \\: \hline or \cline{col1-col2} \cline{col3-col4} ...
	& {\bf max-min-EF} & {\bf Shapley value} & {\bf KS} & {\bf \lexmaxws} \\
	\hline
	$\delta \le 0$  & $(\frac{1+|\delta|}{2}, \frac{1+|\delta|}{2})$ & $(1,|\delta|)$ & $(1,|\delta|)$ & $(1,|\delta|)$ \\
	$0 \le \delta \le \frac{1}{3}$ & $(\frac{1-\delta}{2}, \frac{1-\delta}{2})$ & $(1 - \delta, 0)$ & $(\frac{1 - \delta}{1 + \delta}, \frac{\delta(1 - \delta)}{1 + \delta})$ & $(1 - 2\delta,\delta)$ \\
	$\frac{1}{3} \le \delta \le 1$ & $(\frac{1-\delta}{2}, \frac{1-\delta}{2})$ & $(1 - \delta, 0)$ & $(\frac{1 - \delta}{1 + \delta}, \frac{\delta(1 - \delta)}{1 + \delta})$ & $(\frac{1-\delta}{2}, \frac{1-\delta}{2})$ \\
	\hline
\end{tabular}

We wish to draw the attention of the reader to the following facts.

When $\delta < 0$, implying that each agent desires a different item, the max-min-ES solution dictates a transfer from agent $A$ to agent $B$. Thus agent $B$ not only gets his most desired item, but also gets paid for taking it.  Agent $A$ is better of in the RP disagreement mechanism, where he gets his most desired item without having to pay agent $B$.

The utility that agent $B$ gets from the Shapley value solution is equal to his expected utility at the disagreement point (the expected output of RP).  This goes against our perception of fairness, in which increase in the general welfare should be shared by all those who contributed to the increase.

Both the Kalai-Smorodinsky bargaining solution and the \lexmaxws\ solution share the increase in welfare among the two agents, though to different extents. We leave it to the reader to decide which of the two does it better. The case $n=2$ is too small to illustrate our main reason for preferring \lexmaxws\ over KS, which is the fact that KS does not satisfy decomposition properties.  This is illustrated in the proof of Proposition~\ref{pro:KS}), by an example where $n=4$.

In Section~\ref{app:nucleolus} we explained how the notion of the nucleolus can be adapted to our WS-core, giving the {\bf nucleolus-WS} mechanism. For the two agents case, nucleolus-WS gives the same solution as \lexmaxws, except in the range $0 < \delta < \frac{1}{2}$, where nucleolus-WS gives the pair of utilities $(1 - \frac{3\delta}{2}, \frac{\delta}{2})$. Interestingly, there is a unique value ($\delta = \frac{1}{3}$) in the range $0 < \delta < 1$ for which the KS solution and the nucleolus-WS solution coincide, and in that range KS never coincides with any of the other solution concepts presented above.

\section{Population and resource monotonicity}
\label{sec:monotonicity}

We consider here the unit-demand matching setting. {\em  Population monotonicity} means that by introducing an additional agent, it cannot be that a different agent gains utility. {\em  Resource monotonicity} means that by introducing an additional item, it cannot be that an agent looses utility. Moulin [Econometrica 1992] showed that the Shapley value satisfies both population and resource monotonicity.  Here is an example showing the the lexmax-WS solution does not satisfy population monotonicity and does not satisfy resource monotonicity. In the example both {\em random priority} and the {\em Eating} mechanism give the same disagreement utility, and hence the example applies to both.

\begin{tabular}{|c|c|c|c|c|}
  \hline
  % after \\: \hline or \cline{col1-col2} \cline{col3-col4} ...
  Example 5 & {\bf Item A} & {\bf Item B} & {\bf Item C} & {\bf Item D} \\
\hline
  {\bf Agent 1} & 12 & 0 & 6 & 0 \\
  {\bf Agent 2} & 12 & 6 & 0 & 0 \\
  {\bf Agent 3} & 24 & 12 & 0 & 25 \\
  \hline
\end{tabular}

There are three agents $\{1,2,3\}$ and four items $\{A,B,C,D\}$. The valuation function of agents is as described in the table.
%a vector of length~3, with $v_1 = (12, 0 ,6)$, $v_2 = (12, 6, 0)$ and $v_3 = (24, 12, 0)$.

If only agents $\{1,2\}$ participate and only items $\{A,B,C\}$, then for each agent, both the disagreement utility and the lexmax-WS utility are 9. If the set of agents is changed to $\{1,2,3\}$, the disagreement utilities become~8 for agent~1, only 7 for agent~2, and~14 for agent~3. The lexmax-WS solution has welfare 36, giving agent~1 utility of~9.5 and agent~2 utility of~8.5 (at this point the anticore constraint for the set $\{1,2\}$ is tight). Hence population monotonicity for agent~1 does not hold upon introducing agent~3.

Changing now the set of items to $\{A,B,C,D\}$ %with value~25 for agent~3 and no value to the other agents,
changes the disagreement utilities of each of agents~1 and~2 to be 9. This is also their utility in the respective lexmax-WS solution, hence item monotonicity does not hold -- introducing item $D$ caused the utility of agent~1 to drop from 9.5 to 9.

\input{app-algorithms}

\section{Full information upon request}
\label{sec:discuss}

In this work we assume (as is often assumed in literature on cooperative games and in systems like Spliddit \cite{GP14}) that knowledge of the true valuation functions of the agents is available to our mechanisms. This should not be interpreted as if we assume that all agents know the valuation functions of all other agents. Rather, the interpretation is that a mechanism can request information from the agents about their valuation functions (e.g., the most preferred alternative from a set of alternatives in the {\em random priority} mechanism, full ordinal preferences for the {\em eating} mechanism, cardinal valuations for \lexmaxws) and obtain truthful replies.  In practice, presumably agents will be asked to report their valuation functions to the mechanism in private -- in our mechanisms there is no need for an agent to know the valuation functions of other agents.

Below we explain why we view it as both necessary and reasonable to assume that agents will supply truthful information to our mechanisms.

%The main points of the discussion are as follows. In general, we have no control on whether agents will be truthful or not. All we can attempt to do is to design the mechanism in such a way that it provides incentives for the agents to be truthful (but we do not have guarantees that the agents will respond to these incentives). The most common incentive for truthfulness in game theory literature is to make being truthful a dominant strategy. However, this is impossible to achieve in our setting.

{\bf Dominant strategies.} Informally, a (deterministic) mechanism has  {\em dominant strategies}, if for every agent, whenever the agent needs to provide information to the mechanism, then the combination of the information previously provided to the agent by the mechanism and the valuation function of the agent itself suffice in order for the agent to figure out a ``best response". Here the response is the information that the agent provides, and being ``best" means that regardless of the information not available to the agent (such as valuation functions of other agents), no other response will lead to higher end utility for the agent. (The definition for randomized mechanisms is somewhat more complicated, but not needed for the discussion here.) It is well known (and easy to prove) that in our setting, even in the special case of \rentalprob\ with only two agents, there is no mechanism that satisfies the following three properties simultaneously: maximizing welfare, being budget balanced, and having dominant strategies. Hence it is unavoidable to give up at least one of three properties, and we choose to give up having dominant strategies.

%{\bf On valuation functions.} As is common in related literature, we assume that agents have well defined valuation functions, and that they can compute them. However, this is an abstraction of reality, and typically holds only in some approximate sense. For example, the true value of a room for a student may depend also on factors such as the average temperature in the upcoming summer, or on whether she will break a leg in her upcoming ski vacation. The student may not be able to reliably incorporate such unpredictable factors into her estimate for the value of the room. For such reasons, agents often do not really have dominant strategies even in mechanisms in which they supposedly do.

{\bf Why would agents report their true valuation function to the mechanism?} As remarked earlier, our mechanisms do not have the property that being truthful is a dominant strategy (as this is theoretically impossible). So why is it reasonable to assume that agents will be truthful? We propose here several practical reasons why this may be (approximately) the case in some real world settings. One reason is that games are not played in isolation, but in a larger social context that involves various educational processes and social norms, and this context may encourage truthful behavior in the game. For example, a common social norm is that cheating by someone who has been treated unfairly is more socially acceptable than cheating by someone who was treated fairly. Our solution concept incorporates fairness features (an agent is guaranteed utility at least as high as the disagreement utility, and is furthermore guaranteed that her monetary payments are used only so as to compensate those agents for which the chosen alternative is less desirable, and not so as to provide agents with profits beyond what they could obtain from their best alternative), and this may help reduce the drive to cheat. Another reason why agents may report their true valuation functions is because in our mechanisms, being truthful is a strategy that is not dominated by any other strategy. Unless an agent knows the valuation function of other agents,
being untruthful might cause the agent to lose utility.

{\bf The burden of reporting valuations.} Two features of our \lexmax\ mechanism alleviate some of the cognitive/computational burden that an agent might suffer when computing what to report to the mechanism. One aspect is the continuity property, and moreover, the small Lipshitz constant. For example, in situations where Theorem~\ref{thm:Lipshitz} applies, an agent may provide the requested information up to an additive error of $\epsilon$ of her choice, and be guaranteed that the effect of this error of her final utility will be a difference of at most $\epsilon$. So an agent that is not sensitive to a difference of $\epsilon$ in her utility can afford to compute only $\epsilon$-approximations to her valuation function. The other aspect that sometimes alleviates the burden of reporting valuations is the issue of decomposability, especially in settings like \rentalprob. The mechanism can be broken into two phases, where in the first phase it suffices to report only ordinal valuations, and in the second phase, cardinal valuations need to be reported only for the component to which the agent ends up belonging, and not for the whole input instance.

\section{Discussion of some other modeling assumptions}

{\bf Quasi-linear utilities.} We assume that the utility functions of the agents are {\em quasi-linear}. It is desirable to limit quasi-linearity assumptions so that they need to hold only in a limited range of values. All solutions in the WS-core satisfy the property that for every agent $i$ her utility lies between $\min_{A}[v_i(A)] = W_{min}(i)$ and $\max_{A}[v_i(A)] = W_{max}(i)$. Hence it suffices for our purposes that for every agent $i$, her utility function is quasi-linear in the range $[W_{min}(i), W_{max}(i)]$.

{\bf The role of monetary transfers.} Monetary transfers are used in our solution concept in order to make it beneficial for all agents to move from the disagreement point to a solution that maximizes welfare. Monetary transfers can be used for other purposes as well, such as taxing those agents that happen to be rich and subsidizing those agents that happen to be poor, but these uses of monetary transfers are beyond the scope of this work, and can be applied (if desired) independently of our solution concept.

{\bf Outcome versus process.} We assume that what the agents care about is the final outcome -- the alternative chosen and the transfers. However, sometimes agents care also about the process by which the outcome was reached. For example, players may derive satisfaction not only from ``winning", but also from a sense ``playing well" (e.g., making clever moves in challenging situation, regardless of the outcome). In our setting, by changing the mechanism (e.g., from RP to \lexmax) the nature of the ``game" changes, and the amount of ``pleasure" (or displeasure) derived from playing the game changes. Aspects of this nature are not captured by our work.

{\bf Valuation functions and fairness.} The mechanisms discussed in this paper are based on either ordinal preferences or cardinal valuations of the agents. We remark that there are studies that suggest that even full knowledge of cardinal valuations is insufficient information if the goal is to achieve a solution that is deemed fair by humans. It turns out (see~\cite{YaariBarHillel}, for example) that depending on additional annotation that is provided for the same cardinal valuations, such as whether the valuation is based on {\em needs}, on {\em preferences} or on {\em beliefs}, humans tend to choose different solutions as being fair.

{\bf Disagreement point for \rentalprob.} We assumed a situation in which the students who are faced with the room allocation problem already rented the apartment, and for this setting we used RP as a disagreement mechanism. One may consider also a situation in which the students are contemplating the possibility of renting the apartment, but have not yet committed to renting it. In this case the problem changes because another alternative is introduced, that of not renting the apartment. The value of this alternative for each student is $\frac{r}{n}$ (where $r$ is the total rent and $n$ is the number of students), because this is the amount of money saved by the student by not renting. It is natural to treat this new alternative as the disagreement point. As our \lexmaxws\ mechanism is sensitive to the choice of disagreement point, this will lead to a difference between the solutions that it proposes in the two settings: the one in which the students already committed to rent the apartment, and the one in which they maintain the outside option of not renting.

%% file: discussion.tex
\section{Discussion of domination and of the anticore}\label{sec:discussion}

In our setting, one alternative among ${\cal{A}}$ {\em must} be selected. Our solution concept postulates that an alternative $A^*$ of maximum welfare is selected, and budget-balanced transfers are used in order to share the welfare fairly among the agents. The reason for choosing a maximum welfare $A^*$ is to create as large as possible pool of welfare to be distributed among the agents. The assumption that utility functions are quasi-linear allows the transfers to redistribute the welfare in an arbitrary way among the agents, regardless of the value of $A^*$ to each of the agents.  The outcome of this welfare distribution is summarized by the utilities of the agents, which combine the values derived from $A^*$ and from the transfer. Hence our setting can be described in the following equivalent way. There is a set $\cal{N}$ of agents, and a given amount of welfare which equals $w_{{\cal{N}},v}(A^*)$. This welfare needs to be shared among the agents, where the share of each agent $i$ is her utility $u_i$, under the condition that $\sum_{i \in {\cal{N}}} u_i = w_{{\cal{N}},v}(A^*)$. The reference context that is available to us is the tuple $v$ of valuation functions, the set ${\cal{A}}$ of alternatives, and a probability distribution $\pi_v$ over ${\cal{A}}$. Our welfare-sharing (WS) core (Definition~\ref{def:WS}) is based on two sets of constraints. Below we justify each set of constraints, and also discuss its relation to other notions in cooperative game theory.

\subsection{Domination}% -- ``it would have been worse"}
\label{sec:domination}

We assume that a probability distribution $\pi_v$ over alternatives ${\cal{A}}$ is given. This $\pi_v$ serves as a {\em reference point} (also referred to as a {\em disagreement point}). Namely, $\pi_v$ represents what the agents intend to do (select one alternative $A \in {\cal{A}}$ according to probability distribution $\pi_v$) if they are restricted to use a (randomized) social choice function without transfers. For example, students allocating rooms in a shared apartment may intend to use the Random Priority mechanism, which is quite simple to implement. We propose to them that they use a mechanism with transfers instead, a mechanism that generates more welfare, but may be more complicated to implement. However, we cannot enforce that the agents switch from the reference point to our mechanism. We can only try to convince them to do so, and moreover, we might need to convince each and every one of them before the switch to our mechanism actually happens.
%In some respects, our mechanism might be inferior to their intended mechanism: our mechanism is more complicated, it involves transfers, requires that students report their full valuation functions (and not just choose a room when its their turn to choose), makes quasi-linear assumptions on their utility functions (Random Priority makes no such assumptions), and may require heavier computational resources.
In order for an individual agent to be convinced to switch, it does not suffice that the utility of other agents will increase by the switch -- we need to guarantee that she herself will get higher utility from the new mechanism (or at least not lose utility). Here we make the assumption that agents are risk neutral, and hence the distribution over utilities that an agent derives from the randomized reference point can be summarized by one number -- the expected utility. Hence we impose the domination constraint $u_i(A^*,p) \ge \sum_{A \in {\cal{A}}} \pi_v(A)v_i(A)$ for every agent $i$.

%Two remarks are in order here. One is that unlike the case for anticore (to be discussed in Section~\ref{sec:uppercore}), in the case of domination, constraints over sets of agents do not provide additional information beyond constraints over individual agents, because the set function describing the expected utility that a set of agents obtains at a disagreement point is a linear function. The other remark is that
We remark that if agents are risk averse (rather than risk neutral) then the utility that they derive from the randomized reference point becomes smaller, making our deterministic mechanisms even more attractive.

We assume that a reference point is given for the original game, but make no assumption regarding how this reference point is chosen. In particular, the reference point, which may be a function of the valuations, need not be a \emph{continuous} function of the valuation functions. Indeed, in many natural cases (reference points derived from {\em Random Priority}, or from the {\em Eating mechanism} of~\cite{BM01}) the reference point is not continuous in the valuation functions.

There are other solution concepts that also require domination over a disagreement point. For example, this is the case for Nash bargaining solution, and for the Kalai-Smorodinsky bargaining solution (see Appendix~\ref{appendix:bargaining}). Often, the disagreement point represents the bargaining power of agents -- the utility that they can obtain by not participating in the mechanism. Hence it is natural that the mechanism needs to offer them at least their disagreement utility, as otherwise they would not participate. Considerations of this sort are referred to as {\em individual rationality} (IR). However, in our context, agents do not have an outside option of not participating. Rather, the disagreement point itself is the outcome of a mechanism that involves all agents, referred to as the disagreement mechanism. Hence it is not clear what utility, if any, an agent will derive by refusing to participate in our mechanism, because this may depend on what the other agents do. Hence we view our domination property as implementing a fairness principle rather than reflecting bargaining power: switching from the disagreement mechanism to a new mechanism with transfers generates extra welfare, and it is ``fair" that this extra welfare be shared by all participants, or at the very least, that no agent suffers a loss in utility.

There are solution concepts that fix a particular reference point (e.g., two possible reference points are considered in~\cite{Moulin92}, leading to notions that the paper refers to as {\em weak} and {\em strong} IR). Our work differs from these works in the sense that our mechanisms can receive an arbitrary reference point as an input parameter, and dominate the given reference point.

\subsection{The anticore} % -- ``it should be no better"}
\label{sec:uppercore}

When agent $i$ receives an in-payment of $p_i$, we wish it to be the case that agent $i$ could justify to others why she deserves such a payment. Such a justification is needed because against every in-payment to one agent there is an equal amount of out-payment from other agents, and these other agents need to be convinced that their out-payments are extracted for a good reason. A justification that agent $i$ can provide is that $A^*$ is not her preferred alternative, and so she should be compensated for not contesting the choice of $A^*$. This justification has a limit. Without transfers, the highest utility that agent $i$ can hope to achieve is $W_{max}(i) = \max_{A \in {\cal{A}}} v_i(A)$. Hence there is no justification for $p_i$ to exceed $W_{max}(i) - v_i(A^*)$.
%Another justification that agent $i$ can provide is that she is contributing towards the generation of welfare, via her valuation function $v_i$. Being a generator of welfare, she deserves to receive utility. However, regardless of the alternative chosen, the contribution of agent $i$ towards the total welfare does not exceed $W_{max}(i)$. Hence her contribution to generating the welfare does not justify giving her in-payments that push her utility above $W_{max}(i)$.

More generally, we require that also every set $S$ of agents (e.g., all agents of a certain gender) will be able to justify receiving net in-payment into $S$. Against every in-payment to $S$ there is an equal amount of out-payment from the set $\bar{S} = {\cal{N}} \setminus S$ of remaining agents (e.g., all agents of the other gender), and this other set needs to be convinced that their out-payments are extracted for a good reason.
%Regardless of the alternative chosen, the total contribution of $S$ towards generating welfare cannot exceed $W_{max}(S)$. Likewise,
Without transfers, there is no alternative that can offer $S$ total utility higher than $W_{max}(S)$. Hence it is difficult to justify extracting out-payments from $\bar{S}$ if their use is to increase the utility of $S$ beyond $W_{max}(S)$. This gives the constraints of the anticore.

{\bf Related notions:}
The constraint $u_i \le W_{max}(i)$ is referred to as {\em reasonable from above} (REAB) by Milnor~\cite{Milnor52} (see page 21 in~\cite{PS2007}). %(Likewise, the constraint $u_i(A^*,p) \ge W_{min}(i)$ for the lower core is referred to as {\em reasonable from below}.)
The anticore extends the REAB constraint to sets, requiring $u_S(A^*,p) \le W_{max}(S)$ for every set $S$. The anticore was defined in~\cite{Moulin92}, where the constraint $u_S(A^*,p) \le W_{max}(S)$ was referred to as the {\em stand alone test} for set $S$. The term anticore also appeared in some other work, though not necessarily with the same interpretation. Let us elaborate on this.

We first recall the notion of {\em core} in {\em transferable utility cooperative games}. Suppose that there is a set $\cal{N}$ of players and a {\em characteristic function} $f: {\cal{N}} \rightarrow R$, specifying for each coalition of players the payoff that the coalition can achieve on its own. An {\em imputation} $I = (I_1, \ldots, I_n)$ is a vector of payoffs, distributing the payoff $f({\cal{N}})$ of the grand coalition among the players. Namely, $\sum_{i \in {\cal{N}}} I_i = f({\cal{N}})$. An imputation is said to be in the {\em core} (which we shall call here the {\em imputation core}, so as to distinguish it from others notions of core used in this paper) if for every set $S \subseteq{\cal{N}}$ of players, $\sum_{i \in S} I_i \ge f(S)$.

In our context, the utilities $u_i$ play the same mathematical role as the imputations $I_i$. The function $W_{max}$ indirectly plays the role of the characteristic function, by defining a characteristic function $D(S) = W_{max}({\cal{N}}) - W_{max}({\cal{N}} \setminus S)$ (technically referred to as the {\em dual} of $W_{max}$) and requiring $u_S(A^*,v) \ge D(S)$. This dual definition is more in line with our setting being that of a profit game (the inequalities ensure some minimum utility for each set of agents). However, we find it inconvenient to work with the dual definition (in particular when it is combined with the domination constraints), and hence we use the equivalent definition of $u_S(A^*,v) \le W_{max}(S)$.

The anticore resembles the notion of core in {\em cost-sharing games}, and indeed such a core is sometimes referred to as an anticore (see~\cite{MSS92}). Suppose that there is a set $\cal{N}$ of players and a {\em characteristic function} $f: {\cal{N}} \rightarrow R$, specifying for each coalition of players the cost of receiving service on its own. A {\em cost-sharing vector} $c = (c_1, \ldots, c_n)$ is a vector of costs, distributing the cost $f({\cal{N}})$ of the grand coalition among the players. Namely, $\sum_{i \in {\cal{N}}} c_i = f({\cal{N}})$. A cost-sharing vector is said to be in the {\em anticore} (which we shall call here the {\em cost-sharing anticore}, so as to distinguish it from our notion of anticore) if for every set $S \subseteq{\cal{N}}$ of players, $\sum_{i \in S} c_i \le f(S)$. Our anticore has the same mathematical structure as a cost-sharing anticore, equating $W_{max}$ with the characteristic function $f$, and the utilities $u_i(A^*,p)$ with the cost shares $c_i$. A major difference is that we impose this mathematical structure in a {\em profit-sharing} game in which agents wish to maximize their utility, rather than in a {\em cost-sharing} game in which agents wish to minimize their cost.

The term anticore has been used in~\cite{Derks2014} for profit-sharing games. There, the characteristic function $f: {\cal{N}} \rightarrow R$ specifies for each coalition of players the profit that it can generate on its own. A {\em profit-sharing vector} $u = (u_1, \ldots, u_n)$ is a vector of profits, distributing the profit $f({\cal{N}})$ of the grand coalition among the players. Namely, $\sum_{i \in {\cal{N}}} u_i = f({\cal{N}})$. In~\cite{Derks2014}, a profit-sharing vector is said to be in the {\em anticore}
%(which we shall call here the {\em profit sharing anticore}, so as to distinguish it from our notion of anticore)
if for every set $S \subseteq{\cal{N}}$ of players, $\sum_{i \in S} u_i \le f(S)$. Derks et al~\cite{Derks2014} provide the following justification for the anticore: ``if one coalition obtains less than its worth, then it is only fair that all coalitions obtain at most their worths". Our anticore has the same mathematical structure as that of~~\cite{Derks2014}, equating $W_{max}$ with the characteristic function $f$, and the utilities $u_i(A^*,p)$ with the profit shares $u_i$. However, our function $W_{max}$ does \emph{not} carry the same interpretation that one usually associates with a characteristic function of a profit-sharing game. In our setting, a set $S$ of players \emph{cannot} generate for itself a utility $W_{max}(S)$ if it breaks from the grand coalition. Rather, $W_{max}(S)$ is the maximum utility that $S$ can obtain (without transfers) if it stays in the grand coalition, and all players, including the players not in $S$,
agree to choose the alternative that maximizes the welfare of $S$.

Depending on the nature of the characteristic function $f$, the anticore in~~\cite{Derks2014} might be either empty or nonempty. (In fact, the main content of that work concerns how to handle cases in which both the core and the anticore are empty.) In contrast, in our setting, the anticore is always nonempty. In particular, not using any transfers (equivalently, using the all~0 transfer vector) is always in the anticore. By nature of its construction, our function $W_{max}$, when viewed as a characteristic function in a profit-sharing game, guarantees non-emptiness of the anticore.

%% file: decomposability-app.tex
\section{Decomposability}\label{app:decomp}

The following proposition shows that the goals of maximizing welfare and of satisfying decomposition properties are compatible with each other.

\begin{proposition}
\label{lem:welfarecomponent}
Any alternative that maximizes welfare also maximizes welfare for each component separately.
\end{proposition}

\begin{proof}
Let $A$ be an alternative that maximizes welfare, namely, for which $\sum_{i\in {\cal{N}}} v_i(A)$ is largest possible. For a component $S$, let $B$ be an alternative that maximizes the welfare of $S$, namely, for which $\sum_{i\in S} v_i(B)$ is largest possible. We need to show that $\sum_{i\in S} v_i(A) = \sum_{i\in S} v_i(B)$.

Suppose for the sake of contradiction that $\sum_{i\in S} v_i(A) < \sum_{i\in S} v_i(B)$. Let $C$ be an alternative that maximizes the welfare of $\bar{S}$. Then necessarily  $\sum_{i\in \bar{S}} v_i(C) \ge \sum_{i\in \bar{S}} v_i(A)$. By the fact that $S$ is a component, there must be an alternative $D$ for which $\sum_{i\in S} v_i(D) = \sum_{i\in S} v_i(A)$ and $\sum_{i\in \bar{S}} v_i(D) = \sum_{i\in \bar{S}} v_i(C)$. It follows that $\sum_{i\in {\cal{N}}} v_i(D) > \sum_{i\in {\cal{N}}} v_i(A)$, contradicting the assumption that $A$ maximizes welfare.
\end{proof}

We next prove that strong decomposability implies weak decomposability.
\begin{proposition}\label{prop:stong-decomp-implies-weak}
	Let $M$ be a mechanism that for every instance selects an alternative that maximizes welfare and a budget balanced vector of transfers. If $M$ is strongly decomposable, then for every decomposable instance the solution produced by $M$ is weakly decomposable.
\end{proposition}

\begin{proof}
	Let $S$ be a component. Let $A$ be the alternative (maximizing welfare for $\cal{N}$) chosen by $M$ in instance $I$, and let $A_S$ be the alternative (maximizing welfare for $S$) chosen by $M$ in instance $I_S$. By Proposition~\ref{lem:welfarecomponent} we have that $\sum_{i\in S} v_i(A) = \sum_{i\in S} v_i(A_S)$. By strong decomposability of $M$ we have that $u_i(M(I)) = u_i(M(I_S))$ for every $i\in S$, and consequently $\sum_{i\in S} u_i(M(I)) = \sum_{i\in S}u_i(M(I_S))$. As a utility of an agent is the sum of value for the chosen alternative and the transfer, we have that the sum of transfers of the agents in $S$ must be the same in $I$ and in $I_S$. But in $I_S$ the sum of transfers is~0, because of budget balance. Hence in $I$ the net transfer into $S$ is~0 as well, as required by weak decomposability.
\end{proof}

%\section{Components of a decomposition form a lattice}
%\label{sec:lattice}

We now prove Proposition~\ref{lem:component}, which is restated here for convenience.

\begin{proposition}
%\label{lem:component}
Given a set $\cal{A}$ of alternatives, a set $\cal{N}$ of agents, and a tuple $v$ of valuation functions, the components of $\cal{N}$ form a lattice.
\end{proposition}

\begin{proof}
For every component $S$, also its complement $\bar{S} = {\cal{N}} \setminus S$ is a component as well. Consequently, if we prove that for every pair of components their intersection is also a component, this will imply that their union is a component as well.

Let $S$ and $T$ be components. We need to show that $S \cap T$ is a component.
Let $A$ be an alternative that is Pareto optimal for $S \cap T$. Let the vectors of values that $A$ gives to the sets $S \cup T$, $S \setminus T$, $T \setminus S$ and $\cal{N} \setminus (S \cup T)$ be $a_1, a_2, a_3$ and $a_4$, respectively. Let $B$ be an alternative that is Pareto optimal for ${\cal{N}} \setminus (S \cap T)$. Let the vectors of values that $B$ gives to the sets $S \cup T$, $S \setminus T$, $T \setminus S$ and $\cal{N} \setminus (S \cup T)$ be $b_1, b_2, b_3$ and $b_4$, respectively. We need to show the existence of an alternative $C$ for which the respective vectors of values are $a_1, b_2, b_3, b_4$.

Let $A_S$ be an alternative that is Pareto optimal for $S$ and for every agent in $S$ offers at least as much value as $A$ does. Such an alternative must exist for the following reason. If $A$ is Pareto optimal for $S$, then take $A_S = A$. If $A$ is not Pareto optimal for $S$, then there must be an alternative $A'$ that dominates $A$ for all agents in $S$. Continue the argument with $A'$.

Likewise, let $B_{\bar{S}}$ be an alternative that is Pareto optimal for $\bar{S}$ and for every agent in $\bar{S}$ offers at least as much value as $B$ does. Such an alternative must exist as well. By the fact that $S$ is a component, there must exist an alternative $C_S$ whose vector of values dominates $a_1, a_2, b_3, b_4$.

Repeating the above argument with $T$ instead of $S$, there must exist an alternative $C_T$ whose vector of values dominates $a_1, b_2, a_3, b_4$.

Now take an alternative $A'_T$ that is Pareto optimal for $T$ and for every agent in $T$ offers at least as much value as $C_S$ does, and an alternative $B'_{\bar{T}}$ that is Pareto optimal for $\bar{T}$ and for every agent in $\bar{T}$ offers at least as much value as $C_T$ does. By the fact that $T$ is a component, there must be an alternative $C$ whose vector of values dominates $a_1, b_2, b_3, b_4$. In fact, it must equal $a_1, b_2, b_3, b_4$, because otherwise either $A$ was not Pareto optimal for $S \cap T$, or $B$ was not Pareto optimal for ${\cal{N}} \setminus (S \cap T)$. The existence of $C$ shows that $S \cap T$ is a component, as desired.
\end{proof}

%% file: main-thm-proof.tex
\subsection{Proof of Theorem~\ref{thm:BS}}
%\subsection{Proof of Theorem~\ref{thm:BS}}
\label{sec:main-thm-proof}

Recall that the welfare-sharing (WS) core contains those solutions that are in the anticore and dominate the reference point. We now turn to prove Theorem~\ref{thm:BS}, that the WS-core is nonempty when either $W_{max}$ is sobmodular, or $W_{max} - W_{\pi}$ is monotone. Our proof is based on the well known Bondareva-Shapley theorem~\cite{Bondareva63,Shapley67}, though our setting involves some subtleties that might be overlooked in attempts to directly apply the Bondareva-Shapley theorem.

We first recall some standard terminology.
A set function $f$ is {\em additive} if for every set $S$, $f(S) = \sum_{i\in S} f(i)$. A set function $f$ is {\em submodular} if for every two sets $S$ and $T$ it holds that $f(S) + f(T) \ge f(S \cap T) + f(S \cup T)$. Equivalently, $f$ is submodular if it has the decreasing marginal returns property:  for every item $i$ and two sets $S \subset T$ it holds that $f(S\cup\{i\}) - f(S)  \ge f(T \cup \{i\}) - f(T)$.

A collection of sets $T_1, \ldots, T_k$ and nonnegative coefficients $\lambda_1, \ldots, \lambda_k$ is said to be a {\em fractional cover} for set $S$ if for every item $i \in S$ it holds that $\sum_{j \; | \; i \in T_j} \lambda_j \ge 1$. This fractional cover for $S$ will be referred to as {\em proper} if furthermore $T_i \subset S$ for every $1 \le i \le k$.
%It is an {\em exact fractional cover} if for every item $i \in S$ it holds that $\sum_{j \; | \; i \in T_j} \lambda_j = 1$.
A set function $f$ is (proper) {\em fractionally subadditive} if for every $S$, whenever $T_1, \ldots, T_k$ and $\lambda_1, \ldots, \lambda_k$ form a (proper, respectively) fractional cover for $S$, the inequality $\sum_{j = 1}^k \lambda_j f(T_j) \ge f(S)$ holds.
%A set function $f$ is {\em fractionally superadditive} if its negation $-f$ is a fractionally subadditive set function.

%Functions that are fractional superadditive rather than fractional subadditive are often referred to as {\em totally balanced}. A function is referred to as {\em balanced} if somewhat relaxed conditions hold:  $\sum_{j = 1}^k \lambda_j f(T_j) \le f(S)$ is required to hold only when $S = {\cal{N}}$, and moreover, only when $T_1, \ldots, T_k$ with the coefficients $\lambda_1, \ldots, \lambda_k$ are an exact fractional cover of ${\cal{N}}$.

%The Bondareva-Shapley theorem is as follows:

%\begin{theorem}
%\label{thm:BShapley}
%The imputation core is nonempty if and only if the characteristic function $f$ is balanced.
%\end{theorem}

A set function $f$ is {\em XOS}~\cite{LLN} if for some $t$ there are additive set functions $g_1, \ldots, g_t$ such that for every set $S$, $f(S) = \max_{j=1}^t g_j(S)$.  Observe that by definition, the function $W_{max}$ belongs to the class XOS.

Given $\pi = \pi_v$ (the previously defined reference point), the set function $W_{\pi}(S) = \sum_{A \in {\cal{A}}} \pi(A) \sum_{i \in S} v_i(A)$ describes the expected total utility that a set of agents receives from the reference point.
Observe that $W_{\pi}$ is an additive set function. We now define a characteristic function $f$ as $f(S) = W_{max}(S) - W_{\pi}(S)$. This function belongs to the class XOS (because $W_{max}$ is in XOS and $W_{\pi}$ is additive) and is nonnegative.

It was shown in~\cite{feige} that for nonnegative monotone set functions, the class of XOS functions is the same as the class of fractionally subadditive functions. However, this correspondence need not apply to the XOS function $f$ defined above, because $f$ need not be monotone. The example provided after Definition~\ref{def:WS}   can serve to illustrate this difficulty. For this reason Theorem~\ref{thm:BS} considers two special cases. In the second of these special cases $f(S)$ is indeed monotone, and hence as explained above, it is fractionally subadditive (which implies that it is also proper fractionally subadditive). In the first special case, $W_{max}$ is submodular, and for this case we can use the following proposition.

\begin{proposition}
\label{pro:proper}
If $f$ is a nonnegative submodular function then it is proper fractionally subadditive.
%defined over the set ${\cal{N}}$. Then $f$ (even if it is not monotone) satisfies those fractional subadditive constraints in the special case in each each set $T_i$ that participates in a fractional cover of a set $S$ is required to be a subset of $S$ ($T_i \subset S$).
\end{proposition}

\begin{proof}
Suppose for the sake of contradiction that the proposition is false. Then there is a counter example: a set $S$, a fractional cover for $S$ composed of sets $\{T_i\}$ with nonnegative coefficients $\{\lambda_i\}$, such that
$\sum \lambda_i f(T_i) < f(S)$, and $T_i \subset S$ for all $i$ (this is the property of being a {\em proper} fractional cover). Consider the smallest such counter example, with smallest $|S|$, and conditioned on $|S|$, with smallest $\sum \lambda_i$. For each item $j \in S$, let $c_j = \sum_{i | j \in T_i} \lambda_i$. By virtue of being a fractional cover we have that $c_j \ge 1$ for all $j \in S$.

We claim that there must be at least one item $j$ in $S$ for which $c_j = 1$. This holds because otherwise we could divide all $\lambda_i$ by $\min_j c_j$ and obtain a smaller counter example ($\sum \lambda_i$ would decrease, we will still have a fractional cover, and $\sum \lambda_i f(T_i)$ will decrease because all terms are nonnegative).

Given that there is an item $j \in S$ with $c_j = 1$, we may remove $j$ both from $S$ and from all sets $T_i$, and we remain with a proper fractional cover. The condition $\sum \lambda_i f(T_i) < f(S)$ together with the decreasing marginal returns property imply that  $\sum \lambda_i f(T_i - \{j\}) < f(S - \{j\})$. Hence we have a smaller counter example, which is a contradiction.
\end{proof}

%Recall from Section~\ref{sec:uppercore} that we may define a characteristic function $D(S) = W_{max}({\cal{N}}) - W_{max}({\cal{N}} \setminus S)$. Then we may replace the uppercore constraints by constraints $u_S(A^*,p) \ge D(S)$. As $D(S)$ is balanced, Theorem~{thm:BShapley} implies that the corresponding imputation core is nonempty, and hence welfare can be shared without violating the upper core constraints. However, these welfare shares need not be in the WS-core,  because the utilities that agents get might not satisfy the domination constraints.

%We wish to show that the WS-core is nonempty. Given that $W_{max}$ belongs to the class XOS, the Bondareva-Shapley theorem can be used in a black-box fashion to infer that the anticore is nonempty (a fact that we already knew). We need to further show that for every reference point $\pi$ there is a solution in the anticore that also dominates the reference point $\pi$.

\begin{corollary}
\label{cor:proper}
If the function $W_{max}$ is submodular, then the function $f(S) = W_{max}(S) - W_{\pi}(S)$ as defined above is nonnegative and proper fractionally subadditive.
\end{corollary}

\begin{proof}
As $W_{max}(S)$ is submodular and $W_{\pi}(S)$ is additive, their difference $f(S) = W_{max}(S) - W_{\pi}(S)$ is submodular. The function $f(S)$ is nonnegative because $W_{\pi}(S)$ is the expected welfare that $S$ derives from a distribution over the alternatives, whereas $W_{max}(S)$ is the welfare derived from the best alternative. The proper fractional subadditivity property follows from Proposition~\ref{pro:proper}.
\end{proof}

We are now ready to prove Theorem~\ref{thm:BS}.

\begin{proof}(of Theorem~\ref{thm:BS})
Observe that the theorem is equivalent to the statement that the cost-sharing core with respect to the characteristic function $f(S) = W_{max}(S) - W_{\pi}(S)$ contains a nonnegative cost-sharing vector $c = (c_1, \ldots, c_n)$, with $\sum_i c_i = f({\cal{N}})$. Being in the cost-sharing core with respect to $f$ implies being in the anticore
with respect to $W_{max}$, and being nonnegative implies domination over $W_{\pi}$.

Consider the {\em primal} linear program (LP) with variables $x_1, \ldots, x_n$ (which in a feasible solution will give the desired cost shares):
\medskip

{\bf Maximize} $\sum_{i \in {\cal{N}}}  x_i$ subject to:

\begin{itemize}

\item
$\sum_{i \in S} x_i \le f(S)$ for every set $S$

\item
$x_i \ge 0$ for every $i$

\end{itemize}

The dual to the above LP has variables $y_S$ for every set $S$:
\medskip

{\bf Minimize} $\sum_{S \subseteq {\cal{N}}} f(S) y_S$ subject to:

\begin{itemize}

\item
$\sum_{\{S \: | \; i \in S\}} y_S \ge 1$ for every $i$

\item
$y_S \ge 0$

\end{itemize}

The fact that $f$ is proper fractionally subadditive implies that the value of the dual is at least $f({\cal{N}})$. Being a minimization problem, the value of the dual is in fact exactly $f({\cal{N}})$, by taking $y_{\cal{N}} = 1$, and $y_S = 0$ for $S \not= {\cal{N}}$. Hence the primal LP is feasible and has value $f({\cal{N}})$, as desired.
\end{proof}

%% file: Lorenz-domination-proof.tex
\section{Proof of Theorem \ref{thm:DR89} -- Lorenz domination}\label{app:proof-thm-dr89}
%\subsection{Proof of Theorem \ref{thm:DR89} -- Lorenz domination}\label{app:proof-thm-dr89}

In this section we prove Theorem \ref{thm:DR89}: If $W_{max}$ is submodular, then the \lexmaxws\ solution (which is in the WS-core) % WS-core has a solution that
Lorenz-dominates all other  solutions in the WS-core.

\begin{proof}
Let  $f$ be the characteristic function associated with the WS-core. Namely, $f = W_{max}$, when the valuation function of each agent is such that the expected value of the disagreement point is~0 (recall that this can be enforced by applying an additive shift to the valuation functions). Recall that $f$ is nonnegative, though it need not be monotone. By our assumption, $f$ is submodular. We need to show that the core of the corresponding cost-sharing game contains a solution that Lorenz-dominates all other  solutions in the core.

Let $x = (x_1,x_2, \ldots, x_n)$ denote the \lexmaxws\ solution, as found by the water filling algorithm. We now show that $x \ge_{Lor} y$ for every $y$ in the core. We need to sort the coordinates of $x$ by order of nondecreasing values (equivalently, sort the agents by nondecreasing utilities, above the disagreement point, under $x$). In this order the sets $S_i$ appear in the order $S_1, S_2, \ldots S_m$. More precisely, denoting $S_j \setminus \bigcup_{i < j} S_i$ by $T_j$ for every $j \ge 1$, the sets $T_i$ appear in order $T_1, T_2, \ldots, T_m$. For $y$, we may pick an order of our choice (not necessarily nondecreasing) on its coordinates. If in this arbitrary order we have $\sum_{j \le i} x_j \ge \sum_{j \le i} y_j$ for every $i$, then $x \ge_{Lor} y$. The order we choose for $y$ is also $T_1, T_2, \ldots T_m$, and within each $T_i$ we sort the coordinates by order of nondecreasing values. Observe that for every $N_i = \cup_{j \le i} T_i$ it holds that $x(N_i) \ge y(N_i)$, because by the proof of Proposition~\ref{pro:water} (replacing $\cal{N}$ by $N_i$ in the proof of the proposition), $x$ gives the set $N_i$ its maximum possible value given the constraints of $f$ (and $y$ has to be in the core). Consider now any intermediate coordinate $k$ within an interval $S_i$ (starting at $N_{i-1}$ and ending at $N_i$). At the left endpoint of the interval $S_i$ we have $x(N_{i-1}) \ge y(N_{i-1})$.  The $x$ values are constant within the set $S_i$ whereas the $y$ values are increasing. Hence if it would happen that $\sum_{j \le k} x_j \ge \sum_{j \le k} y_j$ fails to hold, this would lead to $x(N_i) < y(N_i)$, which would be a contradiction. Hence it must be that $x \ge_{Lor} y$.
\end{proof}

%% file: sec-rent.tex
\section{Shared Rental (Unit-Demand Matching)}
\label{sec:rental}
A motivating example for our new solution concept is the \emph{\rentalprob}. %the \emph{rental harmony problem}.
In the \rentalprob\  $n\ge 2$ agents rent an apartment with $n$ bedrooms and jointly need to pay the rent $r$ and decide on the matching of rooms to the agents. We assume that they are already committed to rent the apartment\footnote{We will not handle the question of how they have entered this commitment and if it was rational for them to do so. }
and that each will first pay an equal share of the rent, that is, will pay $r/n$, and then they will need to decide on the matching of the agents to the rooms (each getting exactly one room), and the transfers between the agents to make the outcome ''fair''.
The \rentalprob\
% rental harmony problem
can be restated more abstractly as a matching problem for unit-demand agents which we describe next.

In the problem of \emph{matching with unit-demand agents}, the goal is to allocate a set $\cal{M}$ of $n\geq 2$ items to $n$ unit-demand agents, giving each of them one item.
Each unit-demand agent $i$ assigns a value $v_i(j)$ to each item $j\in \cal{M}$. We would like to match agents with items - each agent will receive exactly one item.\footnote{In our framework no items can be left unallocated.}
Mapping this to the general framework, the set of alternatives is the set of permutations over the $n$ items, and the value that an agent assigns to an alternative is his value for the item he receives in that permutation. That is, permutation $\sigma$ assigns item $\sigma(i)$ to agent $i$, and his value for $\sigma$ is $v_i(\sigma(i))$.

Note that we can indeed frame the \rentalprob\ %rental harmony problem
as a matching with unit-demand agents problem, although in the  \rentalprob\ %rental harmony
there is the additional component of paying the rent. The rent payment from each agent can be encoded by some shifted unit-demand valuations obtained from the original valuations by decreasing the valuation of each item\footnote{We make no assumption that the value of an item is necessarily non-negative. Note that even if that was the case before the shift, it might fail to hold after the shift.} by the rent amount $r/n$. Thus, the two problems are essentially equivalent and we will go back and forth between the two.

\subsection{The WS-core and the \lexmaxws\ solution}
\label{sec:ws-core-lexmax}

To define the welfare-sharing core (WS-core), we first discuss reasonable distributions that the agents might consider as their reference point and expect the solution to dominate. Such distributions naturally arise from some standard mechanisms  (without transfers) that may be considered natural, for example, for the \rentalprob.

\begin{enumerate}
	
	\item {\em Uniform} ($U$). The allocation $\pi$ is a permutation chosen uniformly at random, independent of $v$.
	
	\item {\em Random priority} ($RP$). One selects a random permutation over agents, and then the agents in turn each select one item from those remaining.
	
	\item {\em Eating} ($Eat$), also known as {\em probabilistic serial}~\cite{BM01}. 	Each item has unit volume. Each agent ``eats'' items at the same rate, starting at his most preferred item. Whenever an item is totally consumed, each agent eating it switches to his highest priority item that still has some volume left. When all items are consumed, we have a fractional allocation. It is decomposed into a weighted sum of integral allocations, and one of them is chosen (with probability proportional to its weight).
	
	% 	\item {\em Maximum welfare} ($MW$). Choose an allocation that maximizes welfare (breaking ties at random), with no transfers.
	
\end{enumerate}
It is easy to see that \emph{Uniform} might actually have Pareto dominated allocations in its support, so it seems less attractive to use.
In fact, it is shown in \cite{BM01} that \emph{Eating} Pareto dominated \emph{RP} which in turn Pareto dominates \emph{Uniform}.
As \emph{Uniform} might select an allocation that is Pareto dominated, we take  \emph{RP} as a natural reference mechanism (yet all our claims below will also hold for the \emph{Eating} mechanism). The disagreement point we take will be the expected utility of the output of the RP mechanism on the specific valuations.
We note that this utility is not continuous in valuations when the preference \emph{order} of some agent changes, but it is easy to see it is continuous (and even Lipschitz continuous with a small constant of 1) as long as the ordinal preferences remain the same.
Observe that we can again shift the unit-demand valuations by decreasing from each agent valuation his expected utility from the reference distribution, normalizing his utility in the new disagreement point to zero. We make this normalization assumption in the rest of the section.

To complete the definition of the welfare-sharing core, we consider the anticore constraints $W_{max}(\cdot)$ for these unit-demand valuations.
We observe that in the matching with unit-demand setting, $W_{max}(S)$, the maximal achievable welfare that a set $S\subseteq \cal{N}$ of agents can obtain, is the maximal weighted matching of the set $S$ to any subset $T\subseteq \cal{M}$ with $|T|=|S|$.
The set function $W_{max}(\cdot)$ (when the sets are sets of agents) satisfies the OXS property defined by \cite{LLN}.
They have proved that any such set function is submodular.
%\mbedit{Moreover, computing it for any set of agents is can be done in polynomial time as this is a weighted matching problem.}
We immediately get the following implications by Theorem \ref{thm:BS}, Theorem \ref{thm:DR89} and Proposition \ref{pro:unique}:

\begin{proposition}
	\label{prop:unit-demand-ours}
	For any unit-demand setting, the set function $W_{max}: \cal{N}\rightarrow \Re$ is submodular. Thus, within the WS-core there is a Lorenz dominating solution. This solution is unique, and moreover, it coincides both with the lexicographically-maximal solution (\lexmaxws) and with the min-square solution. % \mbedit{Moreover, this \lexmaxws\ solution can be computed in polynomial time.}
\end{proposition}

So we see that all three above solutions exist and coincide for the unit-demand case, and recall that we call it the \lexmaxws\  solution.
We next observe that in the \rentalprob\ %rental harmony setting
this solution satisfies some desirable decomposability properties.
% Recall that we suggest picking this solution as the fair distribution of the welfare,  and will call the solution the min-square solution. \mbcomment{we have never named our solution. Do we want to name it?}

%\mbedit{% We note that the submodularity of unit-demand valuation is \emph{not} sufficient to prove that $w_{max}$ is submodular.
\begin{remark}
We note that unit-demand valuations are submodular, and $W_{max}$ is submodular for unit-demand valuations. Another case where $W_{max}$ is submodular is when the valuation functions of the agents are additive over the items (which is also a class of submodular valuation functions).
However, there are submodular valuations for which $W_{max}$ is not submodular. See Appendix \ref{app:other-comb} for more details.
\end{remark}
%ask whether $W_{max}$ is submodular for \emph{any} submodular valuations. We
%show that this conjecture is false,
%the submodularity of the valuations is \emph{not} sufficient to imply that $w_{max}$ is submodular,
%by presenting some submodular valuations for which $W_{max}$ fails to be submodular.
%}

\subsection{Decomposability}
\label{sec:decomp}

Recall the notion of decomposability from Definition~\ref{def:decompose}. This notion, when specialized to
of unit-demand matching instances (such as the Shared Rental problem) is equivalent to the definition below.

\begin{definition}
We say that a unit-demand matching instance $I = ({\cal{N}}, {\cal{M}}, v)$ is {\em decomposable} if there is some partition of $\cal{N}$ into $P_1, P_2, \ldots, P_t$ (with $t > 1$) and of $\cal{M} $ into $M_1, \ldots, M_t$, with the properties that $|P_{\ell}| = |M_{\ell}|$ for all $1 \le \ell \le t$, and that in every Pareto optimal allocation, in every part $P_{\ell}$ (referred to also as a {\em component}), each agent of $P_{\ell}$ receives an item from $M_{\ell}$.% higher than every item not in $M_{\ell}$.
\end{definition}

A sufficient condition for $P_1, P_2, \ldots, P_t$ and $M_1, \ldots, M_t$ to serve as a decomposition of an instance is that for every $1 \le \ell \le t$, every agent in $P_{\ell}$ prefers every item in $M_{\ell}$ over every item not in $M_{\ell}$.

Consider a large house with five rooms, three on the east wing and two on the west wing, and five renters. Hence for the room assignment problem, which is a {\em unit demand matching instance}, there are $5! = 120$ alternatives, one for each permutation. Assume that the five renters can be partitioned to two groups, the ``east group" with three renters and the ``west group" with two renters. The agents have the following preferences: each agent in the east group prefers any room on the east wing over any on the west wing,
and each agent on the west group prefers any room on the west over any on the east.
In such a case, given a mechanism that provides solutions to the room assignment problem, there are two natural options regarding how to use it. One is to apply the mechanism on the whole input instance. The other is to first assign each group of agents to its preferred wing, and thereafter apply the mechanism to each wing independently (one such subproblem has three agents and $3! = 6$ alternatives, the other  has two agents and $2! = 2$ alternatives), without further exchange of information or transfer of money between the two groups. If the mechanism enjoys the strong decomposability property (Definition~\ref{def:strongdecompose}) then every agent is indifferent regarding which of the two options is used, in the sense that both options give her the same utility. So the agents may as well decompose the instance.

Conversely, if the strong decomposability property fails, then for some agents the first option gives higher utility than the second, and for some other agents the second option gives higher utility than the first (this will necessarily happen for a budget balanced mechanism that maximizes welfare, because the sum of utilities in both options is the same). Hence it might be difficult to reach agreement among the agents regarding which option to choose. Likewise, had we started with two separate instances, one for the east group and one for the west group, with no group interested in rooms in the other wing, we would be faced with the question of whether or not to compose these two instances into one larger instance for the whole house, as this affects the distribution of welfare among the agents. The use of a strongly decomposable mechanisms eliminates the source of such conflicts.

For unit demand matching instances, decomposition goes beyond the weak and strong decomposition properties of Section~\ref{sec:decompose}. The added feature is that when an instance decomposes, then also each (Pareto optimal) alternative by itself can also be decomposed. For example, an allocation of rooms to agents in the above example can be decomposed into the allocation of the east wing rooms and the allocation of the west wing rooms. In these settings, the fact that a solution decomposes is not only a statement about the end result of the solution, but also about the physical procedure by which the solution can be obtained. We can indeed partition the players into disjoint groups, let each group solve its own subproblem with no communication with the other groups, and the concatenation of the separate solutions derived independently by each group gives back one welfare maximizing alternative (and a vector of budget balanced transfers).

Considering allocation mechanisms without transfers for unit-demand matching, we observe that \emph{Uniform} does not respect the component structure (agents need not get an item from their own part in the partition), whereas both \emph{RP} and \emph{Eating} are strongly decomposable.

When addressing strong decomposability of our \lexmaxws\ mechanism, we postulate that the reference point used by the \lexmaxws\ comes from a strongly decomposable reference mechanism (like \emph{RP}). The reference point for each part of the decomposition
%should \emph{not} be the fixed distribution that we get from running a reference mechanism  on the entire instance, but rather it should be
is the outcome of the execution of the reference mechanism on the corresponding sub-problem separately. The following Corollary is a special case of Proposition~\ref{pro:strongdecompose}.

\begin{corollary}
The \lexmaxws\ mechanism for unit-demand matching problems is strongly decomposable whenever the reference mechanism is strongly decomposable. In particular, this holds when the reference mechanism is either RP or Eating.
% and dominates the outcome of the reference mechanism on each part of the instance decomposition.
\end{corollary}

%\begin{proof}
%	As $W_{max}$ is submodular, by Corollary \ref{cor:submodular} the \lexmaxws\ solution is the same as the unique solution that Lorenz-dominates all other  solutions in the WS-core. The \lexmaxws\ solution can be computed using the water filling algorithm of Section~\ref{sec:select} (or alternatively, by the algorithm in the proof of Theorem \ref{thm:DR89}). It remains to show that the outcome of this algorithm is fully decomposable.
	
%	Let $P_1, P_2, \ldots, P_t$ be the partition of $\cal{N}$ in a decomposable instance. Then $W_{max}$ decomposes into $W_{max}(S) = \sum_{i=1}^t W_{max}(S \cap P_i)$. Applying this decomposition of $W_{max}$ to the description of the water filling algorithm decomposes the algorithm into $t$ algorithms (of the same form) that run independently on the $t$ parts.  Hence the outcome for the whole instance is identical to the concatenation of outcomes for each of the parts, as desired. Further details are omitted.
%\end{proof}

\subsection{Comparison to other solutions}
\label{sec:compare-rental}

It will be insightful to compare our solution to two standard solutions from the literature: the % max-min
envy-free solution, and the Shapley solution (defined in Section \ref{sec:shapley}).
In this section we show that neither one of them satisfies all the properties we are after, even for the \rentalprob.
We also present a complete comparison of the solutions for the case that $n=2$ in Appendix \ref{app:2-agents}.

\subsubsection{Envy-free solutions}\label{sec:Envy-free solutions}

Recall that an allocation is envy free if no agent prefers some other agent's allocation and payment over her own. We have already seen in Section~\ref{sec:background} examples for unit demand matching in which every envy free solution is not in the anticore, does not  dominate \emph{RP}, and does not satisfy decomposability.
%To circumvent the issue of selection from multiple envy-free solutions, we will present (bad) instances with a unique envy-free solution.

%We say that some property (e.g., envy free) is satisfied by a mechanism, if for every instance the mechanism outputs a solution that satisfies the property.

%\begin{proposition}\label{prop:EF-main-problems}
%	Any mechanism for unit-demand matching that always picks an envy-free solution (if exists) does not
%	always output a solution in the anticore,  does not  dominate \emph{RP} (and hence does not dominate Eating either), and
%	does not satisfy decomposability.
%\end{proposition}
Additional related examples are provided by Proposition \ref{prop:EF-problems-app} in Appendix \ref{app:EF-prop}.
In that appendix we also prove the following: 
%%We also present (see Appendix \ref{sec-PRO:RENT})
%Moreover, %in Appendix \ref{sec-PRO:RENT}
%we  present instances of the \rentalprob\ % rental harmony problem
% in which every envy-free solution ends up giving to some agent his most desired room \emph{and} paying him more than the rent he was charged!
%\shortversion{
\begin{proposition}
	\label{PRO:RENT}
	There are instances of the \rentalprob\ % rental harmony
	in which the valuation function of each player is nonnegative and sums up to the total rent (in particular, this is the setting studied in~\cite{GMPZ}), but nevertheless there is a player that in every envy-free solution both gets his most desired room and receives more money than his rent share.
	% There are instances of the \rentalprob\ % rental harmony  	with nonnegative valuation functions in which there is a player that in every envy-free allocation gets his most desired room and instead of contributing towards the rent, the player receives money.
\end{proposition}

%The proof of the proposition appears in Appendix \ref{sec-PRO:RENT}.
%}
\subsubsection{The Shapley value solution}

The Shapley value solution is unique. The next proposition lists some of its properties, see Appendix \ref{app:Shapley-prop} for the proof.
%We next examine it properties. % the properties of the Shapley value solution. %  \mbcomment{point to the section in which it is defined!}

\begin{proposition}\label{prop:Shapley-prop}
	The Shapley-value solution is in the anticore of the \rentalprob\ and satisfies strong decomposability, yet it does not dominate \emph{RP} (and hence also does not dominate Eating).
\end{proposition}

%% file: app-algorithms.tex
\section{Algorithms and computational complexity}
\label{sec:algorithms}

%We consider the following {\em algorithmic template} for computing the \lexmax\ solution. It proceeds in iterations. Initially (at iteration~0), all agents are {\em free} and every agent $i$ starts with her disagreement utility $u_{\pi}(i)$. If any of the constraints of the anticore are tight (satisfied with equality) by this initial solution, then the agents involved in the tight constraints become {\em locked}. Thereafter, in every iteration $j \ge 1$ we do the following. If there are no free agents, the algorithm ends and outputs the utilities of the agents. If there are free agents, then the utility of every free agent is incremented by the same value $x_j$, where $x_j > 0$ is the smallest value that leads to some new anticore constraint becoming tight (equivalently, $x_j$ is the largest increase that does not violate any of the anticore constraints). At this point, all agents involved in a newly tight constraint become locked, and the iteration ends.

%The above algorithmic template has at most $n$ iterations, because in every iteration at least one more agent becomes locked. For an agent $i$ that becomes locked at iteration $k$, her final utility is $u_{\pi}(i) + \sum_{j=1}^k x_j$.

The water filling algorithm of Section~\ref{sec:select} can be seen to imply the following proposition.

\begin{proposition}
\label{pro:rational}
Suppose that the valuation functions of the players are expressed as rational numbers (namely, every $v_i(j)$ is expressed as $\frac{p}{q}$ for some integers $p$ and $q$.) If $W_{max}$ is submodular, the transfers of the \lexmax\ solution are also rational.
\end{proposition}

%We remark that in the special case that $W_{max}$ is submodular, the water filling algorithm has a more efficient implementation (presented in the proof of Theorem~\ref{thm:DR89}) in which the number of anticore constraints that need to be checked decreases in every iteration.

In this section we make the convention that valuation functions take only integer values in the range $[-M, \ldots, M]$, where $M$ is taken to be sufficiently large. (If valuation functions take arbitrary rational values, they can be scaled to give integer values, by multiplying by the lowest common denominator.) As to the output of the algorithm, we shall not insist on getting the exact \lexmax\ solution, but rather a solution that gives every agent a utility that differs by at most $\epsilon$ from her utility in the \lexmax\ solution. Here $\epsilon > 0$ is some parameter of the algorithm that corresponds to the smallest unit of money that agents care about (note that if $M$ is very large, it may well be that $\epsilon > 1$, in which case we assume that $\epsilon$ is an integer). Consequently, the numerical values manipulated by algorithms can be restricted to numbers expressible by $O(\log M + \log (1 + \frac{1}{\epsilon}))$ bits, even though the true \lexmax\ solution might require much higher precision. We view the use of $\epsilon$ as justified in essentially all practical situations, as payments can practically be made only at a precision determined by the smallest denomination accepted in the relevant currency.

In general, running times of algorithms can be expressed as functions of the number of players $n$, number of alternatives $|{\cal{A}}|$, range of valuations $M$, the precision parameter $\epsilon$, and the description length of the disagreement distribution $\pi_v$. To reduce the number of parameters in Definition~\ref{def:complexity}, we do not state explicitly the dependency on $|{\cal{A}}|$ and $\pi_v$. Later, in contexts in which it matters, we will also consider the effect on $|{\cal{A}}|$ and $\pi_v$.

\begin{definition}
\label{def:complexity}
Using the above notation, we consider the following classes of running times for algorithms:

\begin{itemize}

\item {\em Weakly polynomial.}  The number of operations performed is polynomial in $(n, M, \frac{1}{\epsilon})$.

    \item {\em Polynomial.}  The number of operations performed is polynomial in $(n, \log M, \log \frac{1}{\epsilon})$.

    \item {\em Strongly polynomial.} The number of operations performed is polynomial in $n$  and independent of $M$ and $\frac{1}{\epsilon}$, though each operation may involve numbers with $O(\log M + \log (1 + \frac{1}{\epsilon}))$ bits, and hence the time per operation (for example, adding two numbers) might be polynomial in $\log M + \log (1 + \frac{1}{\epsilon})$.

\end{itemize}
\end{definition}

To extract an algorithm out of the water filling algorithm, one needs subroutines for the following three tasks:

\begin{enumerate}

\item Compute the disagreement utilities $u_{\pi_v}(i)$ for every agent $i$.

\item Compute the increments $x_j$ for every iteration $j$.

\item Determine at each iteration which agents are involved in constraints that become tight, so as to lock these agents.

\end{enumerate}

Suppose first that the disagreement utilities are easy to compute. A specific case when this happens is when there is one designated {\em disagreement alternative}, and $\pi$ is supported only on this alternative. In fact, whenever the disagreement utilities are easy to compute, we may add a ``dummy" alternative (which will serve as the disagreement alternative) whose value to each agent exactly equals the computed disagreement utility of the agent, and shift $\pi$ to be supported only on the dummy alternative. (Note that adding this dummy alternative does not change any of the constraints of the anticore.) Hence Theorem~\ref{thm:complexity} extends to all cases in which the disagreement utilities are easy to compute.

\begin{theorem}
\label{thm:complexity}
Consider instances in which the number of alternatives is bounded by some polynomial in $n$, and one is given a disagreement alternative that forms the support of disagreement distribution $\pi$. Then:

\begin{enumerate}
\item If there is even a weakly polynomial time algorithm for computing the \lexmax\ solution on such instances, then $P=NP$.

\item Nevertheless, if $W_{max}$ is submodular, then \lexmax\ can be computed in strongly polynomial time. Moreover, this result extends also to the case where the number of alternatives is not bounded by a polynomial in $n$, provided that there is a {\em value oracle} for computing $W_{max}$ (namely, given $S$, the value of $W_{max}(S)$ can be computed in time polynomial in $n$).
\end{enumerate}
\end{theorem}

\begin{proof}
We first prove the NP-hardness result. It is by reduction from the maximum independent set problem MIS. Let $\alpha(G)$ denote the maximum size of an independent set in a graph $G$. We shall consider the standard gap version $MIS_{c,s}$, where $c$ is the {\em completeness} parameter, $s$ is the {\em soundness} parameter, and $0 < s < c <1$. The input to $MIS_{c,s}$ is a graph $G$ on $n$ vertices, and the computational task is to output {\em yes} if $\alpha(G) \ge cn$, to output {\em no} if $\alpha(G) \le sn$, and any output is allowed if $sn < \alpha(G) < cn$. We may assume without loss of generality that $G$ has no isolated vertices. It is known that $MIS_{c,s}$ is NP-hard for some values of $0 < s < c < 1$ (this is a consequence of the PCP theorem~\cite{ALMSS}, and following a long line of work, the current best bounds can be found in~\cite{KMS18}). This easily implies also NP-hardness when $s = \frac{1}{2}$ for some $c > \frac{1}{2}$. (For example, starting with $MIS_{c,s}$ with $s < 1/2$, add to the graph $(1 - 2s)n$ isolated vertices, and connected them all to the same vertex in the graph.)

We reduce an instance of $MIS_{c,\frac{1}{2}}$ with $\frac{1}{2} < c < 1$ to an instance of \lexmax\ as follows. Given an input graph $G(V,E)$ with $n$ vertices (an instance of $MIS_{c,\frac{1}{2}}$), every vertex $v\in V$ corresponds to an agent and every edge $(u,v)\in E$ corresponds to an alternative. Every agent $v$ derives value $n$ from each of the alternatives that correspond to the edges incident with $v$, and value~0 from every other alterative. Hence every alternative has welfare exactly $2n$. In addition, there is the disagreement alternative $D$ that gives each agent value $\frac{1}{c}$ (where $c$ is the completeness parameter of the $MIS_{c,\frac{1}{2}}$ instance).

This completes the description of the \lexmax\ instance. Observe that the number of agents is $n$ and that $M = n$. Fix $\epsilon = \frac{2 - \frac{1}{c}}{2} = \frac{2c-1}{2c}$. Hence a weakly polynomial time algorithm for \lexmax\ simply needs to run in time polynomial in $n$.

The amount of welfare offered by the maximum welfare alternative (any of the edges) is $2n$. If the welfare could be distributed evenly over all agents, it would give each agent a utility of~$2$. However, such a solution might not be in the anticore. So let us consider the value of $W_{max}(S)$ for various nonempty subsets $S$ of agents. If $S$ forms an independent set in $G$, then $W_{max}(S) = \max[n,\frac{|S|}{c}]$. Else, $W_{max}(S) = 2n$.

It follows that if $G$ is a {\em no} instance of $MIS_{c,\frac{1}{2}}$, giving each agent a utility of~2 is feasible (it is in the anticore). However, if $G$ is a {\em yes} instance, there is a set $S$ of $cn > n/2$ agents for which $W_{max}(S)=n$. As the disagreement utility is $\frac{1}{c}$, all members of this set $S$ get utility exactly $\frac{1}{c} < 2$. Hence the minimum utility in the \lexmax\ solution is $2$ when $G$ is a {\em no} instance and $\frac{1}{c} < 2$ if $G$ is a {\em yes} instance. Computing \lexmax\ with error smaller than $\frac{2c-1}{2c}$ will allow us to distinguish between these cases. This completes the proof of the NP-hardness result.

We now show that \lexmax\ can be computed in strongly polynomial time when $W_{max}$ is submodular. For this, we explain how each of the three tasks of the water filling algorithm can be computed in strongly polynomial time.

\begin{enumerate}

\item {\em Compute the disagreement utilities $u_{\pi}(i)$ for every agent $i$.} This can be done in strongly polynomial time because the disagreement alternative is given.

\item {\em Compute the increments $x_j$ for every iteration $j$.} This task is a special case of a problem known as {\em line search in submodular polyhedra}, and can be solved in strongly polynomial time, given value oracle access for the underlying submodular function~\cite{Nagano2007,GGJ2017}.

\item {\em Determine at each iteration which agents are involved in constraints that become tight, so as to lock these agents.} Once $x_j$ for iteration $j$ has been computed, the agents involved in tight constraints are those agents whose utility cannot be increased, unless either an anticore constraint is violated, or the utility of some other agent is decreased. Computing the maximum increase of utility of an agent, subject to keeping the utility of other agents unchanged and not violating an anticore constraint, is again a special case of {\em line search in submodular polyhedra}. Going over all agents not locked in previous iterations and determining for which of them the maximum increase is~0 gives us the newly locked agents.

\end{enumerate}

\end{proof}

\section{Continuity of the \lexmax\ solution}
\label{sec:continuity}

In this section we consider continuity properties of the \lexmax\ solution as a function of the cardinal valuations of the agents. First, it is important to note that when the disagreement utilities are a function of the valuations, this function might have discontinuity points. For example, this might happen if the disagreement utilities are computed as the outcome of the Random Priority mechanism, and cardinal valuations change to the extent that ordinal preferences over alternatives also change. At discontinuity points for disagreement utilities we shall not require (and do not expect) that the \lexmax\ solution will be continuous. Hence we shall assume in this section that the disagreement utility is the outcome of some distribution $\pi_v$ over the set ${\cal{A}}$ alternatives, and that this distribution does not change when cardinal valuations of agents change (though the disagreement utility itself might change, due to the change in valuations of the alternatives). Hence for agent $i$ the disagreement utility can be expressed as $E_{A \leftarrow_{\pi_v} {\cal{A}}}[v_i(A)]$. We remark that if the disagreement utilities are computed as the outcome of the Random Priority mechanism, our results hold with respect to changes of cardinal valuations that do not alter the ordinal preferences over the alternatives.

As a convention, we shall apply additive shifts to the valuations so that the disagreement utilities are~0. This is done without loss of generality, because the allocation and the transfers of the \lexmax\ solution remain unchanged when an additive shift is applied to the valuation function of an agent. After these additive shifts, $E_{A \leftarrow_{\pi_v} {\cal{A}}}[v_i(A)] = 0$ for every agent $i$. We let $u_i$ denote the utility of agent $i$ in the \lexmax\ solution.

We now introduce additional notation for the purpose of discussing continuity, and the associated Lipshitz constant. We shall consider the effects of the change of the valuation function of a single agent $i$ (while keeping the valuation functions of all other agents fixed) on the utilities of each of the agents. Let $I$ be an instance, and let $v_i$ be the valuation function of agent $i$ (viewed as a vector in $R^m$ where $m =|{\cal{A}}|$). Let $e \in R^m$ be a modification vector, leading to a new valuation function $v'_i  = v_i + e$, and correspondingly a new instance $I'$. To satisfy our convention that the disagreement utility of agent $i$ is~0, we require the modification vector to satisfy $E_{A \leftarrow_{\pi_v} {\cal{A}}}[e(A)] = 0$, which then implies that also $E_{A \leftarrow_{\pi_v} {\cal{A}}}[v'_i(A)]=0$.

Introduce a parameter $t$ (for {\em time}) that changes gradually from $-1$ to $1$. For $-1 \le t \le 1$, let instance $I_t$ be the instance in which the valuation function of $i$ is $v_i + te$. Hence $I_0 = I$ and $I_1 = I'$. Let $|e|$  denote the maximum difference between entries in vector $e$ (the maximum value minus the minimum value).
%The disagreement utility of player $i$ changes by at most $|t| \cdot |e'|$ from instance $I$ to instance $I_t$, and the disagreement utility of other players does not change at all.
For every set $S$ and every $-1 \le t \le 1$, the change in $W_{max}(S)$ from instance $I$ to instance $I_t$ is at most $|t|\cdot |e|$ if $i \in S$, and there is no change if $i \not\in S$.
%After this change, we are left with the situation that only $W_{max}(S)$ changed (and not the disagreement utility), and also now, the change in $W_{max}(S)$ is at most $|t|\cdot |e|$ (when $i \in S$), and there is no change if $i \not\in S$.

For an arbitrary player $j$ (it can be that $j = i$), let $u_j(t)$ denote the utility of player $j$ under the \lexmax\ solution on instant $I_t$.

\begin{proposition}
\label{pro:continuous}
For instance $I$ and modification vector $e$ as above, for every agent $j$ the associated utility function $u_j(t)$ is continuous in the interval $-1 \le t \le 1$ (as long as the WS-core remains nonempty).
\end{proposition}

\begin{proof}
This is a consequence of the water filling algorithm. Details omitted.
\end{proof}

Having established continuity, we now analyse the associated Lipshitz constant. Let $B$ denote the total welfare of the maximum welfare alternative, and let $b_S$ stand for $W_{max}(S)$. The \lexmax\ solution satisfies the following set of linear constraints that we refer to as the WS-constraints:

\begin{enumerate}

\item $\sum_{i \in S} u_i \le b_S$ for every $S \subset [n]$.

\item $\sum_{i = 1}^n u_i = B$ (where $B = b_{[n]}$).

\item $u_i \ge 0$ for every $1 \le i \le n$.

\end{enumerate}

Given a feasible solution $u$ to the WS-constraints (the \lexmax\ solution is one such solution), we say that a set $S \in [n]$ is {\em tight} if its corresponding constraint is satisfied with equality (namely, $\sum_{i \in S} u_i = b_S$). Observe that the set $[n]$ is always tight, by constraint~2, and we may treat the empty set as tight as well.

A collection $\cal{S}$ of sets is a (distributive) {\em lattice} if for every $S \in {\cal{S}}$ and $T \in {\cal{S}}$ it holds that $S\cap T \in {\cal{S}}$ and $S \cup T \in {\cal{S}}$.

\begin{lemma}
\label{lem:lattice}
If $W_{max}$ is submodular, then for every feasible solution $u$, the collection of tight sets (w.r.t. the WS-constraints) forms a lattice.
\end{lemma}

\begin{proof}
Recall that $b_S = W_{max}(S)$.
Let $S$ and $T$ be two sets that are tight under the feasible solution $u$, and for $Y \subset [n]$ let $U(Y)$ denote the sum of utilities derived in $u$ by the players in set $Y$. The tightness implies that $U(S) = W_{max}(S)$ and $U(T) = W_{max}(T)$. Additivity of $U$ implies that $U(S) + U(T) = U(S\cap T) + U(S\cup T)$. Submodularity of $W_{max}$ implies that $W_{max}(S) + W_{max}(T) \ge W_{max}(S\cap T) + W_{max}(S \cup T)$. Consequently, it must hold that $U(S \cap T) = W_{max}(S \cap T)$ and $U(S \cup T) = W_{max}(S \cup T)$, since for every set $Y\subseteq [n]$ it holds that $U(Y)\leq W_{max}(Y)$.
Namely, both $S \cap T$ and $S \cup T$ are tight.
\end{proof}

\begin{remark}
\label{rem:tight}
If $W_{max}$ is submodular, then given only the collection of sets that are tight for \lexmax\ solution (but not \lexmax\ itself), we can compute \lexmax\ as follows. Process the tight sets in an order consistent with the natural partial order over these sets (a tight set is processed only after all tight sets that it contains are processed). Given a set $S$ in the collection, let $S' \subset S$ denote those variables $u_i \in S$ whose value was already determined by sets previously visited in the partial order. Then every variable in $S \setminus S'$ gets value $\frac{b_S - \sum_{j \in S'} x_j}{|S| - |S'|}$.
\end{remark}

\begin{theorem}
\label{thm:Lipshitz}
When $W_{max}$ is submodular and the disagreement utilities are the outcome of some distribution $\pi_v$ over the alternatives, the \lexmax\ solution (which is continuous in $t$, see Proposition~\ref{pro:continuous}) has Lipshitz constant at most~1.
\end{theorem}

\begin{proof}
We shall refer to a value $-1 \le t \le 1$ as a {\em breakpoint} if at that point (either approaching it from the left of from the right or both) some set that was not tight (with respect to the WS-constraints for \lexmax) becomes tight. There are only finitely many breakpoints.
%(In fact, it suffices for our argument that there are only countably many breakpoints.)
Removing these breakpoints, the interval $-1 \le t \le 1$ breaks into finitely many subintervals (open subintervals, except at the points $t = \pm 1$). By continuity, if we bound the Lipshitz constant in every subinterval, this implies the same bound on the Lipshitz constant for the whole interval.

%Consider now an arbitrary subinterval. By Lemma~\ref{lem:lattice}, the collection $\cal{T}$ of tight sets forms a lattice. As in the proof of Proposition~\ref{pro:rational}, the optimal solution (minimizing $\sum_i (x_i)^2$) can be uniquely determined by following the partial order over the sets. Given a set $S\in {\cal{T}}$, let $S' \subset S$ denote those variables $x_i \in S$ whose value was already determined by sets previously visited in the partial order. Then every variable in $S \setminus S'$ gets value $\frac{b_S - \sum_{j \in S'} x_j}{|S| - |S'|}$.

Consider now an arbitrary subinterval. By Lemma~\ref{lem:lattice}, the collection $\cal{T}$ of tight sets forms a lattice. Given two instances $I_t$ and $I_{t + \epsilon}$ in the subinterval, by how much could the solutions change? Recalling Remark~\ref{rem:tight} (and the notation $S$ and $S'$ from that remark), this entails checking by how much $b_S - \sum_{j \in S'} x_j$ might change (due to the change in $v_i$ between the two instances). If $i \not\in S$ then there is no change. If $i \in S$ but $i \not\in S'$ then $b_S$ changes by at most $\epsilon |e|$ and $\sum_{j \in S'} x_j$ does not change, and hence the change is at most $\epsilon |e|$. If $i \in S'$ then consider the collection of sets $\{T_k\}$ such that for every $k$ it holds that $T_k \in {\cal{T}}$, $T_k \subsetneq S$ and there is no set $T \in {\cal{T}}$ such that $T_k \subsetneq T \subsetneq S$ (the $T_k$ are maximal). By the lattice structure, the sets $T_k$ are disjoint. Without loss of generality, $i \in T_1$.  Then the change of value for each variable in $S \setminus S'$ is exactly $\frac{b_S[I_{t+ \epsilon}] - b_S[I_t] - \sum_k (b_{T_k}[I_{t+ \epsilon}] -b_{T_k}[I_t])}{|S| - |S'|}$. Observe that $b_S[I_{t+ \epsilon}] - b_S[I_t]$ is nonzero only as a result of an alternative changing value for $i$, and likewise for $b_{T_1}[I_{t+ \epsilon}] -b_{T_1}[I_t]$ (and for the rest of the $T_k$ we have that $b_{T_k}[I_{t+ \epsilon}] -b_{T_k}[I_t] = 0$). By the definition of $|e|$, the difference in these changes cannot exceed $\epsilon |e|$. (The change in $b_S[I_{t+ \epsilon}]$ is upper bounded by the change in value of the alternative allocated to $i$ when computing $b_S[I_{t+ \epsilon}]$, and similarly for the change in $b_{T_1}[I_{t+ \epsilon}]$. Hence the difference between these two changes cannot exceed $\epsilon |e|$.) Consequently, we get in all cases a Lipshitz constant of at most~1.
\end{proof}

\begin{remark}
\label{rem:Lipshitz}
If the WS-core is nonempty and $W_{max}$ is not submodular, then the Lipshitz constant of \lexmax\ might depend on $n$, the number of agents. For example, suppose that there are $n$ agents and four alternatives, $A_1, A_2, A_3, A_4$. Let $v_1 = (1,2,0,0)$, $v_2 = \ldots = v_{n-1} = (1,0,6,0)$, and $v_n = (1,0,0,6n)$. Alternative $A_1$ serves as the disagreement alternative, and alternative $A_4$ maximizes welfare (which is $6n$). The function $W_{max}$ is not submodular. In particular, $W_{max}(\{1\}) = 2$, $W_{max}(\{1,2\}) = 6$, $W_{max}(\{1,3\}) = 6$, $W_{max}(\{1,2,3\}) = 12$, showing that $W_{max}(\{1,2\}) + W_{max}(\{1,3\}) < W_{max}(\{1\}) + W_{max}(\{1,2,3\})$. The \lexmax\ solution will allocate utilities $(2, 4, \ldots, 4, 2n+6)$. Changing $v_1$ to $(1,3,0,0)$, \lexmax\ will allocate utilities $(3, \ldots, 3, 3n+3)$. Hence a change of~1 in $v_1$ results in a change of $n-3$ in $u_n$.
% this may be a reason to prefer the min-square mechanism over \lexmax\ when $W_{max}$ is not submodular.
\end{remark}

\section{Computing disagreement utilities}

The water filling algorithm for computing the \lexmax\ solution requires the computation of the disagreement utilities $u_{\pi_v}(i)$. Let us discuss briefly the computational complexity of this task in the special case of the room (item) allocation problem. Suppose that there are $n$ agents and $n$ items, that the valuation functions of the agents for the items are given (where $v_i(j)$ is the value that agent $i$ associates with item $j$, and is an integer with absolute value at most $M$), and one needs to allocate one item to each agent. In our setting, the disagreement utility for agent $i$ in \lexmax\ is the expected utility that agent $i$ derives from some default  allocation mechanism with no transfers. We consider here the three candidate default mechanisms that were presented in Section~\ref{sec:ws-core-lexmax}, sketch how they can be implemented algorithmically, and briefly discuss the modifications employed in these mechanisms to handle situations in which the valuation function of an agent might have ties.

\begin{itemize}

\item {\em Uniform (U)}. The disagreement utility of agent $i$ is $\frac{1}{n}\sum_{j=1}^n v_i(j)$. It can be computed exactly in time polynomial in $n$ and $\log M$.

\item {\em Random priority (RP)}. Suppose first that for every agent, her valuation function has no ties (there is no agent $i$ and items $j \not= j'$ such that $v_i(j) = v_i(j')$). The naive approach for computing the disagreement utilities (exactly) involves considering all $n!$ permutations over the agents, and for each permutation determining which item is received by which agent, based on the ordinal preferences of the agents. This procedure can be implemented in polynomial space (in $n$), but it is not polynomial time, and we (the authors) have no reason to believe that there is an alternative algorithm that does compute the disagreement utilities in polynomial time. (Computation of the Shapley value, which is also defined in terms of all possible permutations, is known to be $\#P$ complete in some settings~\cite{DP94}.)

    The RP mechanism is adapted as follows to allow for valuation functions that have ties. Recall that an alternative $A$ is a matching of items to agents. Given a permutation over the agents (as before, all $n!$ permutations are considered), each agent in her turn is faced with a list of alternatives that are still available, and items that are still available to her (matched to her in an least one of the remaining alternatives). Of the items available to her, the agent selects one or more items as most desirable, and discards those alternatives that match to her an item that is not one of the most desirable available items. Implementing this mechanism naively seems to require $n!$ space to store all alternatives. However, it can also be implemented in polynomial space (though still not polynomial time) as follows.    When it is the turn of an agent $i$ to select an item, and several of the remaining items are tied in being most desirable (all have the highest value under $v_i$), the agent is temporarily put on hold, allowing subsequent agents to select items. At every step, if the set of agents on hold contains a subset $S$ of agents whose union of desirable items is also of size $|S|$ (we call such a set {\em tight}), the members of $S$ each get one of their desired items. Finding tight sets can be done in polynomial time (using standard algorithms for bipartite matching).

    %Such an adaptation has the desirable property that RP always gives solutions that are strongly Pareto efficient: agents cannot exchange items in such a way that some agent gains value but no agent loses value.

    We remark that one can also compute the disagreement utilities up to precision $\epsilon$ using a randomized weakly polynomial time algorithm that succeeds with high probability. This is done by randomly sampling $O(\frac{M^2\sqrt{\log n}}{\epsilon^2})$ permutations, for each of them computing the allocation obtained when agents serially select their most preferred item, and for every agent averaging over the utilities that she derives from all the allocations.

\item {\em Eating mechanism (EAT)}. The Eating mechanism can naturally be adapted to the case that the valuation function of an agent may have ties: instead of ``eating" one item at rate~1, the agent can ``eat" all $k$ tied items, each at rate $1/k$. EAT has at most $n$ phases, where a phase ends when some item becomes fully consumed. The length of each phase can be computed in a number of operations that is polynomial in $n$. The precision required in order to express the exact length of a phase may grow significantly as phases progress. Hence practically one would compute the output up to some desirable precision $\epsilon$. This can be done in strongly polynomial time (details omitted).

\end{itemize}

Combining the above discussion on the Eating mechanism and Theorem~\ref{thm:complexity}, we have the following corollary.

\begin{corollary}
\label{cor:roomalg}
In the room allocation problem with $n$ agents and with the Eating mechanism serving as a disagreement point, the \lexmax\ solution can be computed in strongly polynomial time.
\end{corollary}

We remark that for small values of $n$, the polynomial time algorithms of Theorem~\ref{thm:complexity} and Corollary~\ref{cor:roomalg} might be slower than other simpler to implement algorithms that are not polynomial time. Consider for example the room allocation problem.  If the disagreement point is the random priority mechanism, then the disagreement utility of all agents can be determined by considering all $n!$ orders over players. The anticore has $2^n - 1$ constraints. In each iteration $j$ of the water filling algorithm, one can check what upper bound each of these constraints places on $x_j$, and take for $x_j$ the smallest of these upper bounds. Shared apartments rarely have more than $n=4$ rooms, and the above simple algorithm will run very quickly on such instances.

\ignore{

\begin{appendix}

In continuation to Remark~\ref{rem:Lipshitz}, consider the following system of linear equations (one may think of them as the tight constraints for some optimization problem), with $n \ge 4$ even.

\begin{itemize}

\item $x_1 = b_1$

\item $x_1 + x_2 = b_2$

\item $x_1 + x_3 = b_3$

\item For every $2 \le i \le \frac{n}{2} - 1$:

\begin{itemize}

\item $x_{2i-2} + x_{2i-1} + x_{2i} = b_{2i}$

\item $x_{2i-2} + x_{2i-1} + x_{2i + 1} = b_{2i+1}$

\end{itemize}

\item $x_{n-2} + x_{n-1} + x_n = b_n$.

\end{itemize}

This system has a unique solution. Changing $b_1$ to $b_1 + \epsilon$ changes $x_{2i}$ by $(-1)^{i}2^{i}\epsilon$. Hence for some systems of linear equations with 0/1 coefficients, the Lipshitz constant may be exponential in $n$.
}
\ignore{
\end{appendix}